\documentclass[reqno]{amsart}

\usepackage{amscd,amsmath,amsfonts,amssymb,amsthm,cite,mathrsfs,color}
\usepackage{dsfont} \usepackage{ leftidx}

\newtheorem{theorem}{Theorem}[section]
\newtheorem{lemma}[theorem]{Lemma}
\newtheorem{proposition}[theorem]{Proposition}
\newtheorem{coro}[theorem]{Corollary}
\theoremstyle{definition}
\newtheorem {definition}[theorem]{Definition}
\theoremstyle{remark}
\newtheorem{remark}[theorem]{Remark}
\newtheorem*{acknow}{Acknowledgments}

\numberwithin{equation}{section}
\numberwithin{theorem}{section}

\def\iy{\infty}

\def\be{\begin{equation}}
\def\ee{\end{equation}}
\def\bae{\begin{eqnarray}}
\def\eae{\end{eqnarray}}

\DeclareMathOperator{\Ai}{Ai}

\def\LBE{\mathrm{L}\beta\mathrm{E}}
\def\JBE{\mathrm{J}\beta\mathrm{E}}
\def\HBE{\mathrm{H}\beta\mathrm{E}}
\def\BE{\beta\mathrm{E}}

\begin{document}

\title[Asymptotics for   products of characteristic polynomials]{Asymptotics for products of characteristic polynomials in classical  $\beta$-Ensembles}

\author{Patrick Desrosiers} \address{Instituto Matem\'atica y F\'isica,
Universidad de Talca, 2 Norte 685, Talca,
Chile}\email{Patrick.Desrosiers@inst-mat.utalca.cl}

\author{Dang-Zheng Liu} \address{School of Mathematical Sciences, University of Science and Technology of China, Hefei,
230026, P.R. China \& Instituto Matem\'atica y F\'isica,
Universidad de Talca, 2 Norte 685, Talca,
Chile}\email{dzliu@ustc.edu.cn}

%\date{29 March 2013}
\date{June 2013}



\begin{abstract}
We study the local    properties of eigenvalues for the Hermite (Gaussian), Laguerre (Chiral) and Jacobi $\beta$-ensembles of $N\times N$ random matrices.  More specifically, we calculate scaling  limits of the expectation value of products  of characteristic polynomials as $N\to\infty$.     In the bulk of the spectrum of each $\beta$-ensemble, the same scaling limit is found to be $e^{p_{1}}{}_1F_{1}$  whose exact expansion in terms of Jack polynomials is well known.  The scaling limit at the soft edge of the spectrum for the Hermite and Laguerre $\beta$-ensembles is shown to be a multivariate Airy function, which is defined as a generalized Kontsevich integral. As corollaries, when $\beta$ is  even,   scaling  limits  of the $k$-point correlation functions for the three ensembles are obtained.  The asymptotics of the multivariate Airy function for large and small arguments is also  given.   All the asymptotic results rely on a generalization of Watson's lemma and the steepest descent method for integrals of Selberg type.
\end{abstract}

\maketitle

\small
\tableofcontents
\normalsize

\newpage
\section{Introduction}
\subsection{$\beta$-Ensembles of random matrices}
%$b_{k}(\beta)$
In this article  we  consider three classical $\beta$-ensembles of Random Matrix Theory, namely the  Hermite (Gaussian), Laguerre (Chiral), and  Jacobi $\beta$-ensembles (H$\beta$E, L$\beta$E, and J$\beta$E for short).
Their eigenvalue probability density functions are equal to
\begin{eqnarray}\label{eqbetadensities}
\frac{1}{G_{\beta,N}}\prod_{1\leq i\leq N}e^{-\beta x^{2}_{i}/2} \,\prod_{1\leq j<k\leq
N}|x_{j}-x_{k}|^{\beta}, &\quad x_i\in \mathbb{R}, &\quad\text{H$\beta$E},\label{PDFforgauss}\\
\frac{1}{W_{\lambda_1, \beta,N}}\prod_{1\leq i\leq N}x_{i}^{\lambda_1}e^{-\beta x_{i}/2} \,\prod_{1\leq j<k\leq
N}|x_{j}-x_{k}|^{\beta},  &\quad x_i\in \mathbb{R}_+,& \quad \text{L$\beta$E},\label{PDFforlaguerre}\\
\frac{1}{S_{N}(\lambda_{1},\lambda_{2},\beta/2)}\prod_{1\leq i\leq N}x_{i}^{\lambda_{1}}(1-x_{i})^{\lambda_{2}} \,\prod_{1\leq j<k\leq
N}|x_{j}-x_{k}|^{\beta}, & \quad x_i\in (0,1),& \quad \text{J$\beta$E}.\label{PDFforjacobi}
\end{eqnarray}
The normalization constants are all special cases of Selberg's celebrated formula \cite{forrester} and are given in the appendix.

For special values of the Dyson index $\beta$, we recover more conventional random matrix ensembles \cite{forrester,mehta}. The $\beta=1, 2, 4$--ensembles indeed correspond to the ensembles of random  matrices whose respective probability measures  exhibit orthogonal, unitary or symplectic symmetry.

For general $\beta>0$, Dumitriu and Edelman \cite{due}
constructed tri-diagonal real symmetric matrices with independent entries randomly drawn from some specific distributions  and whose eigenvalues are distributed according to \eqref{PDFforgauss} and \eqref{PDFforlaguerre}.  Killip and Nenciu \cite{KN} later obtained a similar construction  for the J$\beta$E.  These explicit constructions play a key role in connecting the $\beta$-ensembles  with one-dimensional  stochastic differential equations in the limit $N\to \infty$ \cite{es,rr,rrv}.   Many probabilistic quantities of interest such as the global fluctuations, the gap probabilities, and the distribution of the largest eigenvalues  were also studied in the limit $N\to \infty$ (see for instance \cite{dv,due3,kill,rrv,rrz,vv,vv2,caer,FS}).  More recently, some universality results  concerning the general $\beta>0$ case have been obtained \cite{bey,bey2}. It was shown that as $N\to \infty$, the eigenvalues in the bulk (middle) of the spectrum of any $\beta$-ensemble defined by density  $\prod_{i=1}e^{-\beta V(x_i)}\prod_{1\leq i<j\leq N}|x_i-x_j|^\beta$, where $V$ is a real-valued function on $\mathbb{R}$, (thus excluding the circular ensembles) are correlated, when appropriately rescaled and after a small averaging,  as the eigenvalues of the  Hermite $\beta$-ensemble.

Apart from being related to the eigenvalues of random matrices, the densities \eqref{PDFforgauss}--\eqref{PDFforjacobi} have alternative physical interpretations.  Indeed, these densities  appeared very recently in theoretical high energy physics \cite{bempf,cem,mmm,sulkovski}. Moreover,  densities such as  \eqref{PDFforgauss}--\eqref{PDFforjacobi}  are equivalent to the Boltzmann factor for classical log-potential  Coulomb gas and to the ground state wave functions squared for Calogero-Sutherland $N$-body quantum systems of the  type $A_{N-1}$, $B_N$ and $BC_N$.   We refer the reader to Forrester's lectures \cite{forrester1} for more details.  The wave functions of the Calogero-Sutherland models  are typically  written in terms of a very special family  of symmetric polynomials, namely the Jack polynomials \cite{kadell,macdonald,stanley}. This connection between  the $\beta$-ensembles and the Jack polynomials has been exploited by many authors and has shown to be very fruitful \cite{forrester0,nagao,bf, due, due2, due3, duk, ER, ke, df, matsumoto1, des, matsumoto2, matsumoto3, dl,for2012}.
\subsection{Products of characteristic polynomials}
Now let $X$ be an $N\times N$ random matrix in some $\beta$-ensemble.  Our aim is to find exact and explicit expressions for the large $N$ limit of the expectation value of $\prod_{j=1}^n\det(X-s_j)$.  We will thus study the following expectation value:
\be
K_N(s_1,\ldots,s_n)=\Big\langle \prod_{i=1}^N\prod_{j=1}^n (x_i-s_j)\Big\rangle_{x\in \BE}
\ee
where   $x=(x_1,\ldots, x_N)$ denotes the eigenvalues of the random matrix $X$ and where the angle brackets  stand for the expectation value.   More explicitly, for the densities \eqref{PDFforgauss}--\eqref{PDFforjacobi},
\begin{multline}
K_N(s_1,\ldots, s_n)=\\ \frac{1}{Z_{N}}\int\cdots \int  \prod_{i=1}^N\prod_{j=1}^n(x_i-s_j) \exp\Big\{ - \frac{\beta}{2}\sum_{j=1}^{ N}V(x_j)\Big\} \prod_{1\leq i<j\leq
 N}|x_{i}-x_{j}|^{\beta} \,dx_1\cdots dx_N,
\end{multline}
where $Z_{N}$ is some normalization constant and
\be V(x_j)=\begin{cases} x^{2}_{j} & \text{H$\beta$E} \\ x_{j}-(2/\beta)\lambda_{1}\ln x_{j} & \text{L$\beta$E} \\
-(2/\beta)\lambda_{1}\ln x_{j} -(2/\beta)\lambda_{2}\ln (1-x_{j})& \text{L$\beta$E} \end{cases} \ee
Actually, we will see that it is more convenient to consider the weighted quantity
\be \label{weightedquantity}
\varphi_{ N}(s_1,\ldots,s_n)= \exp\Big\{ - \frac{1}{2}\sum_{j=1}^{n}V(s_j)\Big\} \, K_N(s_1,\ldots,s_n).
\ee
% We want to evaluate $K_N(A+Bs_1,\ldots,A+Bs_n)$ as $N\to\infty$ for appropriate choices of $A$ and $B$, which usually depend upon $N$.

At this point, it is worth stressing  that if $\beta$ is even and if we let $n=k\beta$, then $\varphi_N(s_1,\ldots,s_n)$ gives access to the $k$-point correlation function \cite{forrester,mehta}:
\begin{multline} R_{k,N}(x_1,\ldots,x_k)= \\  \frac{(k+N)!}{N!}\frac{1}{Z_{k+N}}\int\cdots \int \exp\Big\{ - \frac{\beta}{2}\sum_{j=1}^{k+N}V(x_j)\Big\} \prod_{1\leq i<j\leq
k+N}|x_{i}-x_{j}|^{\beta} d x_{k+1} \cdots d x_{k+N}
\end{multline}
Indeed,
\be R_{k,N}(x_1,\ldots,x_k)=\frac{(k+N)!}{N!}\frac{Z_{N}}{Z_{k+N}}   \prod_{1\leq i<j\leq
k}(x_{i}-x_{j})^{\beta}\big[\varphi_N(s_1,\ldots,s_n)\big]_{\{s\}\mapsto \{x\}}.
\ee  Here the notation $\{s\}\mapsto \{x\}$ means that the variables $s_i$ are evaluated as follows:
\be s_{(i-1)\beta+j}=x_{i}\quad\text{for}\quad i=1,\ldots,k \quad\text{and}\quad j=1,\ldots,\beta .\ee

For the special values $\beta=1,2,4$,  the scaling limits of products of characteristic polynomials  are already well established  (see \cite{af,bds,bs} and references therein).  Thanks to the orthogonal polynomial method, they can be expressed as determinants or Pfaffians of one-variable special functions and their derivatives.  The functions in question depend on the bulk or edge of the spectrum we are looking at \cite{foredge}.  Close to the hard edge of the spectrum (which correspond to $s_i=0$ for Laguerre and $s_i=0,1$ for Jacobi ensembles), one gets a Bessel function or equivalently a $_0F_1(z)$ function.  In the bulk of the spectrum of each ensemble (for instance, about $s_i=1/2$, $2N$, and $0$ for Jacobi, Laguerre, and Hermite, respectively), the formulas involve trigonometric functions or complex functions of exponential type, such as ${}_0F_0(iz)= {}_1F_1(a;a;iz)$.  Finally, at the soft edge of the spectrum (i.e., about $s_j=4N$ for Laguerre and $s_j=\sqrt{2N}$ for Hermite) the scaling limits contain the Airy function $\Ai$.  In fact, these three regimes of large $N$ asymptotics (Hard-Bessel, Bulk-Trigonometric, Soft-Airy) correspond to the three most common universality classes for ensembles of random matrices with $\beta=1,2,4$ (e.g., see Chapter 7 in  \cite{forrester}).

Much less is known about the general $\beta>0$ case. For the $\HBE$,  Aomoto \cite{aom2} and more recently Su \cite{su}, obtained the limiting expectation value of the product of $n=2$ characteristic polynomials respectively in the bulk and at the soft edge of the spectrum.  When both $n$ and $N$ are finite but arbitrary, Baker and Forrester \cite{bf} proved that $K_N(s_1,\ldots,s_n)$ is either a multivariate Jacobi, Laguerre
or Hermite polynomials with parameter $\alpha=\beta/2$  depending on whether the density considered is \eqref{PDFforjacobi}, \eqref{PDFforlaguerre} or \eqref{PDFforgauss}. They also used the theory of multivariate hypergeometric functions developed by Kaneko \cite{kaneko} and Yan \cite{yan}, to express the expectation value $K_N(s_1,\ldots,s_n)$ in the Jacobi and Laguerre $\beta$-ensembles as an $n$-dimensional integral:
\be \label{genint} K_N(s_1,\ldots,s_n)=C_N\int_\mathscr{C}\ldots \int_\mathscr{C} \prod_{j=1}^n e^{-Np(t_j)} \,\prod_{1\leq j<k\leq n}|t_j-t_k|^{4/\beta} q_N(t;s)dt_1\cdots dt_n,
\ee
where $\mathscr{C}$ and $ q_N(t;s)$ respectively denote the circle in the complex plane (or intervals of real numbers)  and multivariate hypergeometric function whose precise forms depend on the ensemble.

Such passage from an $N$-dimensional integral with Dyson parameter $\beta$ to an $n$-dimensional integral with Dyson parameter $\beta'=4/\beta$ is an example of duality relation, which turns out to be very useful since the second $n$-dimensional integral allows us,  in principle at least, to take the limit $N\to \infty$. As observed in \cite{bf}, when $s_1=\cdots=s_n$, the function $q_N(t;s)$ greatly simplifies.  This was exploited in \cite{df} for determining the asymptotic behavior of the eigenvalue marginal density  $\rho_N(x)$.
% We note also that the  expectation of product of $n$ characteristic polynomials for the Circular $\beta$-Ensemble (a special case of $\JBE$) was studied by Matsumoto \cite{matsumoto}.

Other dualities for the Hermite and Laguerre  $\beta$-Ensembles were obtained in \cite{des}.  One particular duality was used to prove that   as $N\to \infty$, for $\beta=1,2,4$, and for an appropriate choice of $A$ and $B$, the expectation $K_N(A+Bs_1,\ldots,A+Bs_n)$ in the $\HBE$ is proportional to the folowing generalized Airy function (see Section \ref{hypergeometric} for more details about the notation):
\be\label{Airydef} \mathrm{Ai}^{(\beta/2)}(s_1,\ldots,s_n)=\frac{1}{(2\pi)^n}\int_{\mathbb{R}^n}e^{ip_3(w)/3} |\Delta(w)|^{4/\beta}{\phantom{j}}_0\mathcal{F}_0^{(\beta/2)}(s_1,\ldots,s_n;iw) d^{n}w,  \ee
which is absolutely convergent for all $(s_1,\ldots,s_n)\in\mathbb{R}^n$ and $\beta\in\mathbb{R}_+$.   The case $\beta=2$ is proportional  to Kontsevich's matrix Airy function \cite{kon}.

\subsection{Main results}\label{sectionmainresults}

We prove that at the soft edge, the expectation value of products of characteristic polynomials in both $\LBE$ and  $\HBE$ actually lead to the same multivariate Airy function.  Note that the simplest asymptotics for the multivariate Airy function is given in Proposition \ref{propAiry}.

\begin{theorem}[Soft edge expectation] \label{theosoft} %Consider the $\HBE$ and $\LBE$.
Let
\be \label{softAB} A,\, B=\begin{cases}(2N)^{1/2},\; 2^{-1/2}N^{-1/6} & \text{for the $\HBE$}, \\  4N, \qquad \; 2(2N)^{1/3} & \text{for the $\LBE$}. \end{cases}\ee   Then as $N\to \infty$,
\be \label{eqlimitingsoft}   \Phi_{\!^{N,n}}^{-1}  \,\varphi_N(A+Bs_1,\ldots,A+Bs_n)\sim(2\pi)^{n}(\Gamma_{4\!/\beta\!,n})^{-1}\,\mathrm{Ai}^{(\beta/2)}(s_1,\ldots,s_n),
\ee
where the coefficients $\Phi_{\!^{N,n}}$ and $\Gamma_{4\!/\beta\!,n}$ are given in the appendix.
\end{theorem}

The above  theorem suggests that the multivariate Airy function is the universal expectation value at the soft edge.  In other words, for any  $\beta$-ensemble characterized by a potential  $V$,  the average of the  product of $n$ characteristic polynomials, when appropriately rescaled and re-centered at the soft edge, should become independent of $V$  and should be proportional to $\mathrm{Ai}^{(\beta/2)}(s_1,\ldots,s_n)$ as $N\to\infty$.

The indication for universality is even stronger in the bulk of the spectrum.  We indeed find  that the three classical $\beta$-ensembles possess the same asymptotic limit for the weighted expectation value  $\varphi_N$ in the bulk, which turns out to be a multivariate hypergeometric function of exponential type. We only state the result in the case of $n=2m$, for simplicity; for $n=2m-1$,  the combination of $\varphi_N$ and $\varphi_{N-1}$ exhibits a universal pattern, which is given in Theorem \ref{odduniversality}.
\begin{theorem}[Bulk expectation] \label{theobulk} For the $\HBE$, $\LBE$, and $\JBE$, let $A$ be equal to $\sqrt{2N}$, $4N$, and $1$, respectively.   Let
\be
\rho(u)=\begin{cases} \tfrac{2}{\pi}\sqrt{1-u^{2}},\quad  u\in (-1,1), &\HBE, \\
  \tfrac{2}{\pi}\sqrt{\tfrac{1-u}{u}}\phantom{tt},\quad  u\in (0,1), &\LBE, \\
	\tfrac{1}{\pi}\tfrac{1}{\sqrt{u(1-u)}},\quad  u\in (0,1), & \JBE. \end{cases}\ee
%%$$ A,\, \rho,\, I=\begin{cases} \sqrt{2N},\, \tfrac{2}{\pi}\sqrt{1-u^{2}},\, (-1,1) & \HBE\\  4N,\,\tfrac{2}{\pi}\sqrt{\tfrac{1-u}{u}}, (0,1) & \LBE \\  1,\,\tfrac{1}{\pi}\tfrac{1}{\sqrt{u(1-u)}},\, (0,1) & \JBE \end{cases}.$$
 Assume  moreover that  $n=2m$ is even. Then  as $N\to \infty$,
\begin{equation}\label{eqlimitingbulk}\frac{1}{\Psi_{\!^{N,2m}}}\varphi_N\Big(A u+\frac{A s_1}{\rho\!(u) N},\ldots,A u+\frac{A s_n}{\rho\!(u) N}\Big)\,\sim\,
\gamma_{m}\!(\!4\!/\!\beta\!)\, e^{-i\pi p_1(s)}\!\!{\phantom{j}}_1F_1^{(\beta/2)}(2m/\beta;2n/\beta;2i\pi s)\end{equation}
where
$\Psi_{\!^{N,2m}}$ and $\gamma_{m}\!(\!4\!/\!\beta\!)$  respectively stand for the coefficient given in \eqref{eqPsi} and \eqref{univcoefficientevenbulk}.
\end{theorem}

 It's  worth emphasizing that the universal coefficient $\gamma_{m}\!(\!4\!/\!\beta\!)$, when $\beta=2$,
	 is conjectured  to be closely related to the moments of the Riemann's $\zeta$-function \cite{bh,ks}.

The hard edge also involves a single hypergeometric series, which is \!\!${\phantom{j}}_0F^{(\beta/2)}_{1}$.  The latter can be seen as a multivariate Bessel function.  We may thus surmise once more that the asymptotic expectation at the hard edge is universal.

\begin{theorem}[Hard edge expectation]\label{theohard}  Let $B$ be  equal to $ N^{-1}$ and $N^{-2}$ for  the $\LBE$ and $\JBE$, respectively.
%\be B=\begin{cases}N^{-1} & \LBE \\ N^{-2}&\JBE \end{cases} \ee
Then as $N\to\infty$,
\be\label{eqlimitinghard}\frac{1}{\xi_{N,n}}K_N(B s_1,\ldots,B s_n)\sim \!{\phantom{j}}_0 F^{(\beta/2)}_{1}\left((2/\beta)(\lambda_1 +n);-s_1,\ldots,-s_n\right).
\ee
The coefficient $\xi_{N,n}$ is given in \eqref{eqxi}.
\end{theorem}

Before explaining how we prove these theorems, let us give some immediate consequences. As previously mentioned, if $\beta$ is even and if  $n=2m=\beta k$, then  scaling limits of  the $k$-point correlation functions  immediately follow  from  the above theorems.  Since the hard-edge case for the $\LBE$ and $\JBE$ is already  known (see Section 13.2.5 in \cite{forrester}),  we only  display below the  results  at the soft edge and in the bulk.

\begin{coro}[Soft edge correlations]\label{theosoftCF} Assume that $\beta$ is even and  $n=\beta k$.
Let $A$ and $B$ be as in Theorem \ref{theosoft}. Then  as $N\to \infty$ in the $\HBE$ and   $\LBE$,
\be \label{softuniversal}  B^{k}R_{k,N}(A+Bx_1,\ldots,A+Bx_k)\;\sim \;a_{k}(\beta)\; |\Delta(x)|^{\beta} \big[ \mathrm{Ai}^{\!(\beta/2)}(s)\big]_{\{s\}\mapsto \{x\}}\; ,
\ee
where $a_{k}(\beta)$ is given in \eqref{univcoefficienta}.
\end{coro}

\begin{coro}[Bulk correlations]\label{theobulkCF} Assume that $\beta$ is even and  $n=\beta k$.  Let $A$ and $\rho(u)$ be as in Theorem \ref{theobulk}.
%Let $A=\sqrt{2N},\, 4N,\, 1$ and  $\rho=\tfrac{2}{\pi}\sqrt{1-u^{2}}$, $\tfrac{2}{\pi}\sqrt{\tfrac{1-u}{u}}$ and $\tfrac{1}{\pi}\tfrac{1}{\sqrt{u(1-u)}}$ respectively corresponding to the $\HBE$, $\LBE$ and $\JBE$.
Then  as $N\to \infty$ in the $\HBE$,   $\LBE$  and   $\JBE$,
\begin{multline}\label{bulkuniversal}   \Big(\frac{A}{\rho(u) N}\Big)^{k}R_{k,N}\Big(Au+\frac{A x_1}{\rho(u) N},\ldots, Au+\frac{A x_k}{\rho(u) N}\Big)\;\sim \;\\
b_{k}(\beta)\; |\Delta(2\pi x)|^{\beta}\; \Big[e^{-i\pi p_1(s)} \!{\phantom{j}}_1F_1^{(\beta/2)}(n/\beta;2n/\beta;2i\pi s)\Big]_{\{s\}\mapsto \{x\}},
\end{multline}
where $b_{k}(\beta)$ is given in \eqref{univcoefficientb}.
\end{coro}

We stress that the function on the right-hand side of \eqref{bulkuniversal} already appeared in Random Matrix Theory: it is exactly the same as the limiting $k$-point correlation function of the circular $\beta$-ensemble with $\beta$ even, which is equal to the function $\rho_{(k)}^\text{bulk}(x_1,\ldots,x_k)$ given in Proposition 13.2.3 of \cite{forrester}.  Moreover, as mentioned previously,  the limiting $k$-point correlation function in the bulk of the  $\HBE$ was recently shown to be  universal \cite{bey,bey2} .  Note that the latter references do not give however,  the  explicit form  of this universal   $k$-point correlation function.   As a consequence of Corollary  \ref{theobulkCF},   we now know that the universal $k$-point correlation function,   when $\beta$ is even, is equal to the hypergeometric function on the right-hand side of \eqref{bulkuniversal}.

It is worth noting that in \cite{jv}, the point process limit of the Laguerre $\beta$-ensemble in the bulk was shown to be the same
(for a general choice of parameters which may  depend on $N$) as the Hermite bulk limit.  The fact that the point process limits of the Jacobi $\beta$-ensemble at the soft edge and  the hard edge
are the same (for a general choice of parameters) as the corresponding limits for the Laguerre
ensemble was later proved in \cite{hmf}. In these cases our results above for the  Laguerre and Jacobi ensembles hold as well, see subsection \ref{parametervarying}.

To the best of our knowledge, there is still no universality theorem regarding the limiting $k$-point correlation function at the soft edge of the spectrum in $\beta$-ensembles with general potential.  However, assuming that such a universal correlation function exits for $\beta$ even, we see from Corollary \ref{theosoftCF} that  it must be equal to the $k$-variable function on the right-hand side of \eqref{softuniversal}, and as a consequence, it must involve  the multivariate Airy function $\mathrm{Ai}^{\!(\beta/2)}$.

\subsection{Organization of the article and proofs} As explained in Section \ref{hypergeometric},   the multivariate hypergeometric functions of the form ${}_p F^{(\alpha)}_{q}$ are defined as series of Jack symmetric polynomials.  Although apparently complicated, these series can be evaluated efficiently by simple  numerical methods\cite{ke}.

The proof of Theorem \ref{theohard} becomes trivial once we know that the expectations of products of characteristic polynomials in the $\LBE$ and $\JBE$ are respectively given by hypergeometric functions of the form and ${\phantom{|}}_1 F^{(\beta/2)}_{1}$ and ${\phantom{|}}_2 F^{(\beta/2)}_{1}$ .  This is known since \cite{forrester0,bf}.

The proof of  Theorems \ref{theosoft} and \ref{theobulk} is not so simple as that of Theorem \ref{theohard}.   Indeed, the calculation of scaling limits in the bulk and at the soft edge requires the asymptotic evaluation of integrals of Selberg type, such as \eqref{genint}.   The whole Section \ref{laplace} is devoted to this task.  We generalize   the Laplace or steepest descent method for high-dimensional integrals that contain the absolute value of the Vandermonde determinant.

In Section \ref{rmt}, we finally apply these results to our three classical $\beta$-ensembles and demonstrate the theorems previously stated in Section \ref{sectionmainresults}. Note that generalizations of Theorems \ref{theohard} and \ref{theohard}, which treat the cases where the parameters $\lambda_1$ and $\lambda_2$ linearly depend upon $N$,  are given in Section 4.4.   Note also that the explicit  calculations abundantly make use of simple transformations of multivariate hypergeometric series, which will be given in the next section. 

Some remarks on PDEs satisfied by the scaling limits are finally given in Section \ref{pdes}.

%%%%%%%%%%%%%%%%%%%%%%%%%%%%%%copy of sects 2,3
\section{Jack polynomials and hypergeometric functions}
\label{hypergeometric}
This section first provides a brief review of some aspects of symmetric polynomials and especially Jack polynomials.  The classical references  on the subject are Macdonald's book \cite{macdonald}
and Stanley's article \cite{stanley}.  More recent textbook treatments of Jack polynomials can be found in \cite[Chap. 12]{forrester} and \cite{kato}.  This will allow us to introduce  the multivariate hypergeometric functions \cite{kaneko,koranyi,yan}.
A few results proved here will be used later in the article.

\subsection{Partitions}

A partition $\kappa = (\kappa_1,\kappa_2,\ldots,\kappa_i,\ldots)$ is a sequence of non-negative integers $\kappa_i$ such that
\begin{equation*}
    \kappa_1\geq\kappa_2\geq\cdots\geq\kappa_i\geq\cdots
\end{equation*}
and only a finite number of the terms $\kappa_i$ are non-zero. The number of non-zero terms is referred to as the length of $\kappa$, and is denoted $\ell(\kappa)$. We shall not distinguish between two partitions that differ only by a string of zeros. The weight of a partition $\kappa$ is the sum
\begin{equation*}
    |\kappa|:= \kappa_1+\kappa_2+\cdots
\end{equation*}
of its parts, and its diagram is the set of points $(i,j)\in\mathbb{N}^2$ such that $1\leq j\leq\kappa_i$.
Reflection in the diagonal produces the conjugate partition
$\kappa^\prime=(\kappa_1',\kappa_2',\ldots)$.

The set of all partitions of a given weight are partially ordered
by the dominance order: $\kappa\leq \sigma $ if and only if $\sum_{i=1}^k\kappa_i\leq \sum_{i=1}^k \sigma_i$ for all $k$.
One easily verifies that $\kappa\leq\sigma$ if and only if $\sigma'\leq\kappa'$. %We shall also require the inclusion order on the set of
%all partitions, defined by $\sigma\subseteq\kappa$ if and only if $\sigma_i \leq\kappa_i$ for all $i$, or equivalently, if and only if the
%diagram of $\sigma$ is contained in that of $\kappa$.

\subsection{Jack  polynomials}

Let $\Lambda_n(x)$ denote the algebra of symmetric polynomials in $n$ variables $x_1,\ldots,x_n$ and with coefficients in the field $\mathbb{F}$.    In this article,   $\mathbb{F}$    is assumed to be the field of  rational functions in the parameter $\alpha$.  As a ring, $\Lambda_n(x)$ is generated by the power-sums:
\be \label{powersums} p_k(x):=x_1^k+\cdots+x_n^k. \ee
The ring of symmetric polynomials is naturally graded: $\Lambda_n(x)=\oplus_{k\geq 0}\Lambda^k_n(x)$, where $\Lambda^k_n(x)$ denotes the set of homogeneous polynomials of degree $k$.   As a vector space, $\Lambda^{k}_n(x)$ is equal to the span over  $\mathbb{F}$ of all symmetric monomials $m_\kappa(x)$, where $\kappa$ is a partition of weight $k$ and
\be  m_\kappa(x):=x_1^{\kappa_1}\cdots x_n^{\kappa_n}+\text{distinct permutations}.\nonumber
\ee
Note that if the length of the partition $\kappa$ is larger than $n$,  we set $m_\kappa(x)=0$.

The whole ring $\Lambda_n(x)$ is invariant under the action of  homogeneous differential operators related to the Calogero-Sutherland models \cite{bf}:
\be \label{diffops} E_k=\sum_{i=1}^n x_i^k\frac{\partial}{\partial x_i},\qquad D_k=\sum_{i=1}^n x_i^k\frac{\partial^2}{\partial x_i^2}+\frac{2}{\alpha}\sum_{1\leq i\neq j \leq n}\frac{x_i^k}{x_i-x_j}\frac{\partial}{\partial x_i},\qquad k=0,1,\ldots.
\ee
The operators $E_1$ and $D_2$ are special since they also preserve each $\Lambda^k_n(x)$.  They can be used to define the Jack polynomials.  Indeed, for each partition $\kappa$, there exists a unique symmetric polynomial $P^{(\alpha)}_\kappa(x)$ that satisfies the following two conditions \cite{stanley}:
\begin{align} \label{Jacktriang}(1)\qquad&P^{(\alpha)}_\kappa(x)=m_\kappa(x)+\sum_{\mu<\kappa}c_{\kappa\mu}m_\mu(x)&\text{(triangularity)}\\
\label{Jackeigen}(2)\qquad &\Big(D_2-\frac{2}{\alpha}(n-1)E_1\Big)P^{(\alpha)}_\kappa(x)=\epsilon_\kappa P^{(\alpha)}_\kappa(x)&\text{(eigenfunction)}\end{align}
where the coefficients $\epsilon_\kappa$ and  $c_{\kappa\mu}$ belong to  $\mathbb{F}$.  Because of the triangularity condition, $\Lambda_n(x)$ is also equal to the span over $\mathbb{F}$ of all Jack polynomials  $P^{(\alpha)}_\kappa(x)$, with $\kappa$ a partition of length less or equal to $n$.

\subsection{Hypergeometric series}\label{subsecseries}

Recall that the arm-length  and leg-length  of the box $(i,j)$ in the partition $\kappa$  are respectively given by
\begin{equation}\label{lengths}
    a_\kappa(i,j) = \kappa_i-j\qquad\text{and}\qquad l_\kappa(i,j) = \kappa^\prime_j-i.
\end{equation}Then the hook-length of a partition $\kappa$  is   defined as the following product:
\begin{equation} \label{defhook}
    h_{\kappa}^{(\alpha)}=\prod_{(i,j)\in\kappa}\Big(1+a_\kappa(i,j)+\frac{1}{\alpha}l_\kappa(i,j)\Big).
\end{equation}
Closely related is the following $\alpha$-deformation of the Pochhammer symbol:
\begin{equation}\label{defpochhammer}
    [x]^{(\alpha)}_\kappa = \prod_{1\leq i\leq \ell(\kappa)}\Big(x-\frac{i-1}{\alpha}\Big)_{\kappa_i} = \prod_{(i,j)\in\kappa}\Big(x+a^\prime_\kappa(i,j)-\frac{1}{\alpha}l^\prime_\kappa(i,j)\Big)
\end{equation}
In the middle of the last equation,  $(x)_j\equiv x(x+1)\cdots(x+j-1)$ stands for the ordinary Pochhammer symbol, to which $\lbrack x\rbrack^{(\alpha)}_\kappa$ clearly reduces for $\ell(\kappa)=1$.  The right-hand side of \eqref{defpochhammer} involves the co-arm-length  and co-leg-length  of box $(i,j)$ in the partition $\kappa$, which are respectively defined as
\begin{equation}\label{colengths}
    a^\prime_\kappa(i,j) = j-1,\qquad \text{and}\qquad l^\prime_\kappa(i,j) = i-1.
\end{equation}

We are now ready to give the  precise definition of the hypergeometric series used in the article.
\begin{definition}\label{HFpqDef}
Fix $p,q\in\mathbb{N}_0=\{0, 1, 2,\ldots\}$ and let $a_1,\ldots,a_p, b_1,\ldots,b_q$ be complex numbers such that $(i-1)/\alpha-b_j\notin\mathbb{N}_0$ for $i=1,2,\ldots,n$. The $(p,q)$-type hypergeometric series is defined as follows \cite{kaneko,koranyi,yan}:
\begin{equation}\label{hfpq}
  {\phantom{j}}_p F^{(\alpha)}_{q}(a_1,\ldots,a_p;b_1,\ldots,b_q;x) = \sum_{k=0}^{\infty}\sum_{|\kappa|=k} \frac{1}{h_{\kappa}^{(\alpha)}}\frac{\lbrack a_1\rbrack^{(\alpha)}_\kappa\cdots\lbrack a_p\rbrack^{(\alpha)}_\kappa}{\lbrack b_1\rbrack^{(\alpha)}_\kappa
    \cdots\lbrack b_q\rbrack^{(\alpha)}_\kappa}P_{\kappa}^{(\alpha)}(x).
\end{equation}
Similarly the hypergeometric series in two sets of $n$ variables,  $x=(x_1,\ldots,x_n)$ and  $y=(y_1,\ldots,y_n)$, is given in  \cite{yan} by
\begin{equation}\label{hfpq2}
    {\phantom{j}}_p\mathcal{F}^{(\alpha)}_{q}(a_1,\ldots,a_p;b_1,\ldots,b_q;x;y)
                = \sum_{\kappa} \frac{1}{h_{\kappa}^{(\alpha)}}\frac{\lbrack a_1\rbrack^{(\alpha)}_\kappa\cdots\lbrack a_p\rbrack^{(\alpha)}_\kappa}{\lbrack b_1\rbrack^{(\alpha)}_\kappa
    \cdots\lbrack b_q\rbrack^{(\alpha)}_\kappa}\frac{P_{\kappa}^{(\alpha)}(x)P_{\kappa}^{(\alpha)}(y)}{P_{\kappa}^{(\alpha)}(1^{n})},
\end{equation}
        where we have used the shorthand notation $1^{n}$ for $\overbrace{1,\ldots,1}^{n}$.
\end{definition}

Note that when $p\leq q$, the above series converges absolutely for all $x\in\mathbb{C}^n$,  $y\in\mathbb{C}^n$ and  $\alpha\in\mathbb{R}_+$.  In the case where $p=q+1$, then the series converge absolutely for all $\left\| x\right\|<1$, $\left\| y\right\|<1$ and  $\alpha\in\mathbb{R}_+$.  See \cite{kaneko} for more details about  convergence issues.
%%\underline{}Note that another frequently-used notation of Jack polynomial is $$C_{\kappa}^{(\alpha)}(x)=\frac{|\kappa|!}{h_{\kappa}^{(\alpha)}}P_{\kappa}^{(\alpha)}(x).$$ The (generalised) hypergeometric functions were studied, in particular, by Kor\'anyi \cite{koranyi}, Yan \cite{yan}, and Kaneko \cite{kaneko}.
%%In the subsequent sections, we will pay special attention to the hypergeometric series ${}_0F_1, {}_1F_1, {}_2F_1$ and ${}_0\mathcal{F}^{(\alpha)}_0, {}_1\mathcal{F}^{(\alpha)}_0, {}_0\mathcal{F}^{(\alpha)}_1$.

Now we give   some  translation  properties of ${}_0\mathcal{F}^{(\alpha)}_0$ and ${}_1\mathcal{F}^{(\alpha)}_0$, which prove  to be of crucial importance.
  For convenience, we write
        \be (a^n)=(\overbrace{a,\ldots,a}^{n}),\qquad b+a x=(b+a x_{1},  \ldots, b+a x_{n}),\qquad   \frac{x}{1-a x}=
\Big( \frac{x_{1}}{1-a x_{1}}, \ldots, \frac{x_{n}}{1-a x_{n}}\Big) ,\ee
where $a, b$ are complex numbers  and $x=(x_{1},  \ldots, x_{n})$.

\begin{proposition}[Translation formulas]\label{proppractical}
 \begin{align}\displaystyle \label{vip1} &\hspace{-4cm}(1) \,  {\phantom{j}}_0\mathcal{F}^{(\alpha)}_0(a+x;b+y)=\exp\{nab+a p_{1}(y)+b p_{1}(x)\} {\phantom{j}}_0\mathcal{F}^{(\alpha)}_0(x;y).\\
    \label{vip2}  & \hspace{-4cm}(2)\,  {\phantom{j} }_1\mathcal{F}^{(\alpha)}_0(a;b+x;y)=\prod_{j=1}^{n}(1-b y_{j})^{-a} {\phantom{j}}_1\mathcal{F}^{(\alpha)}_0(a;x;\frac{y}{1-b y}).\end{align}
\end{proposition}

\begin{proof}
First set $F(a,b)=\!\!\!{\phantom{j}}_0\mathcal{F}^{(\alpha)}_0(a+x;b+y)$.  Then, according to Eq.\ (3.3) of \cite{bf},
\be E_{0}^{(y)}\!\!\!{\phantom{j}}_0\mathcal{F}^{(\alpha)}_0(x;y)=p_{1}(x)\! {\phantom{j}}_0\mathcal{F}^{(\alpha)}_0(x;y),\ee
so that $$\frac{\partial F}{\partial b}=E_{0}^{(y)}\!\!\!{\phantom{j}}_0\mathcal{F}^{(\alpha)}_0(x;y)\big|_{x\rightarrow a+x,y\rightarrow b+y}=p_{1}(a+x)\,F.$$
Similarly,  $$\frac{\partial F}{\partial a}=p_{1}(b+y)\,F.$$
By solving those differential equations we get
$$F(a,b)=e^{a p_{1}(b+y)}F(0,b)=e^{a p_{1}(b+y)}e^{b p_{1}(x)}F(0,0),$$
which is the first desired result.

For the second result, it suffices to prove
\be\label{toprove}{\phantom{j}}_1\mathcal{F}^{(\alpha)}_0(a;x;y)\prod_{j=1}^{n}(1-b y_{j})^{a}=\!\!\!{\phantom{j}}_1\mathcal{F}^{(\alpha)}_0(a;x-b;\frac{y}{1-b y}).\ee
Now let  $G_{l}(b)$ and  $G_{r}(b)$ respectively denote the left-hand and right-hand sides of \eqref{toprove}.
In  Eq.\ (A.1) of \cite{bf}, substitute  $x$ and $y$ by $c x$ and $ y/b$ respectively, then we let $b,c\rightarrow \iy$
and conclude that\!${\phantom{j}}_1\mathcal{F}^{(\alpha)}_0(a;x;y)$ satisfies
\be \label{10eq}E_{0}^{(x)}F-E_{2}^{(y)}F=a p_{1}(y)F.\ee
From  \eqref{10eq}, we get
$$G'_{r}(b)=\Big(-E_{0}^{(x)}\!\!\!{\phantom{j}}_1\mathcal{F}^{(\alpha)}_0(a;x;y)+E_{2}^{(y)}\!\!\!{\phantom{j}}_1\mathcal{F}^{(\alpha)}_0(a;x;y)\Big)\big|_{x\rightarrow x-b,y\rightarrow \frac{y}{1-b y}}=-ap_{1}(\frac{y}{1-b y})\,G_{r}(b).$$
Finally, it's easy to check that  $G_{l}(b)$ also satisfies the differential equation just given above with the same initial condition at $b=0$, which means $G_{l}(b)=G_{r}(b)$. \end{proof}

\begin{coro}\label{practical} \begin{align}\displaystyle
 {\phantom{j}}_0 \mathcal{F}_{0}^{(\alpha)}&(x_1,\ldots,x_n;a^k, b^{n-k})=e^{b p_1(x)} \!\!\!{\phantom{j}}_1 F^{(\alpha)}_{1}(k/\alpha; n/\alpha; (a-b)x_1,\ldots,(a-b)x_n)\\ \displaystyle
&=e^{(a-b)kx_1+b p_1(x)} \!\!\!{\phantom{j}}_1 F^{(\alpha)}_{1}(k/\alpha; n/\alpha; (a-b)(x_2-x_1),\ldots,(a-b)(x_n-x_1)).
\end{align}\end{coro}
\begin{proof}  Proposition \ref{proppractical} implies that
\be  {\phantom{j}}_0 \mathcal{F}^{(\alpha)}_{0}(x_1,\ldots,x_n;a^{k},b^{n-k})=e^{b p_1(x)}\!\!\!{\phantom{j}}_0 \mathcal{F}^{(\alpha)}_{0}(x_1,\ldots,x_n;(a-b)^{k},0^{n-k}).\nonumber
\ee
The right-hand side is equal to
\be e^{b p_1(x)}\sum_{\ell(\kappa)\leq k} \frac{(a-b)^{|\kappa|}}{h_{\kappa}^{(\alpha)}}\frac{P_{\kappa}^{(\alpha)}(x)P_{\kappa}^{(\alpha)}(1^k)}{P_{\kappa}^{(\alpha)}(1^n)}. \nonumber\ee
But we also know from Eq.\ (10.20) of Chapter VI \cite{macdonald} and Eq.\ \eqref{defpochhammer} above, that
\be P_{\kappa}^{(\alpha)}(1^k)=\frac{\alpha^{|\kappa|}[k/\alpha]_\kappa^{(\alpha)}}{\prod_{(i,j)\in\kappa} (\alpha a_\kappa(i,j)+l_\kappa(i,j)+1)}.\nonumber
\ee
Moreover, one easily checks that $[k/\alpha]_\kappa^{(\alpha)}=0$ whenever $\ell(\kappa)> k$. Consequently,
\be
 {\phantom{j}}_0 \mathcal{F}^{(\alpha)}_{0}(x_1,\ldots,x_n;a^{k},b^{n-k})=e^{b p_1(x)}\sum_{\kappa} \frac{(a-b)^{|\kappa|}}{h_{\kappa}^{(\alpha)}}\frac{[k/\alpha]_\kappa^{(\alpha)}}{[n/\alpha]_\kappa^{(\alpha)}}P_{\kappa}^{(\alpha)}(x), \nonumber
\ee which is equivalent to the first equality.

Likewise, the second equality also follows from Proposition \ref{proppractical}.
\end{proof}

We conclude this section with introduction of the  following symbol, for simplicity,
\be  \label{defE} E_{j}^{(\alpha)}(x):=E_{j}^{(\alpha)}(x_{1}, \ldots,x_n)=\!\!\!{\phantom{j}}_0\mathcal{F}_{0}^{(\alpha)}(x;(-1)^{j},1^{n-j}), \qquad j=0,1,\ldots,n .\ee
These functions possess some similar properties with the exponential function, for instance, $$E_{j}^{(\alpha)}(2i\pi +x)=E_{j}^{(\alpha)}(x).$$

\section{Asymptotic methods for integrals of Selberg type}
\label{laplace}

We want to get the large $N$ asymptotic behaviour for integrals of the form \eqref{genint} by generalizing the  classical steepest descent method, or more generally    Laplace's  method for contour integrals. We refer the reader to  Olver's textbook \cite{olver} for more details on  Laplace's method in the one-dimensional case.
 The Laplace method was previously generalized to the multidimensional case by some authors (e.g., see  \cite{breitung,hsu} and references therein), but their work is not directly applicable to our Selberg type integrals due to the presence of the absolute value of the Vandermonde determinant.  Examples of asymptotic evaluation  of specific Selberg type integrals can be found e.g. in \cite{df,forrester0}.

Before going any further, let us adopt some new notational conventions.  Our variables are  $t=(t_1,\ldots,t_n)$ and the  Vandermonde determinant is given by
\be \Delta(t)=\prod_{1\leq i<j\leq n}(t_{i}-t_{j}).\ee
Moreover, $\lambda=(\lambda_1,\ldots,\lambda_n)$ denotes a sequence of parameters such that  $\lambda_j>0$ for all $j$ while  $k=(k_1,\ldots,k_n)$ denotes a sequence of non-negative integers of weight  $|k|=\sum_{j} k_j$.
Finally,   $t^{\lambda-1}=\prod_{j}t_{j}^{\lambda_j-1}$ and $t^{k}=\prod_{j}t_{j}^{k_j}$.

\subsection{Watson's Lemma}

 We first give the Watson's lemma for multiple integrals with Vandermonde determinants. The process of the  proof will suggest a natural extension to contour integrals.

 \begin{lemma}
 Let $\|t\|=\max\{|t_{1}|,\ldots, |t_{n}|\}$ and $q(t)$ be a function of the positive real variables $t_{j}$, such that
 \be q(t)=t^{\lambda-1}\Big(\sum_{j=0}^{m-1}a_{j}(t)+O(\|t\|^{m})\Big) \quad (\|t\|\rightarrow 0)\ee
 where  $$a_{j}(t)=\sum_{|k|=j}a_{j,k}\,t^{k} $$ is a  homogenous polynomial of degree $j$. Then
  \begin{equation} \label{laguerreintegral}I_N:=\int_{0}^{\iy}\cdots \int_{0}^{\iy}
  e^{-N\sum_{j=1}^{n}t_{j}}|\Delta(t)|^{\beta} q(t)\,d t_{1}\cdots d t_{n} =\sum_{j=0}^{m-1}\frac{A_{j}}{N^{n_{\beta}+|\lambda|+j}}+O(N^{-n_{\beta}-|\lambda|-m})\end{equation}
 as $N\rightarrow \iy$ provided that this integral converges throughout its range for all sufficiently large $N$, where $n_{\beta}=\beta n(n-1)/2$ and
  \be \label{monomialintegral}A_j=\int_{0}^{\iy}\cdots \int_{0}^{\iy}t^{\lambda-1}
  e^{-\sum_{j=1}^{n}t_{j}}|\Delta(t)|^{\beta} a_{j}(t)\,d t_{1}\cdots d t_{n}.\ee
\end{lemma}

\begin{proof}
For each $m$, define
\be \phi_{m}(t)=q(t)-t^{\lambda-1} \sum_{j=0}^{m-1}a_{j}(t).\ee
One obtains \be \label{rightintegral} I_N=\sum_{j=0}^{m-1}\frac{A_{j}}{N^{n_{\beta}+|\lambda|+j}}+\int_{0}^{\iy}\cdots \int_{0}^{\iy}
  e^{-N\sum_{j=1}^{n}t_{j}}|\Delta(t)|^{\beta} \phi_{m}(t)\,d t_{1}\cdots d t_{n}. \ee

As $\|t\|\rightarrow 0$ we have  $\phi_{m}(t)=t^{\lambda-1} O(\|t\|^{m})$. This means that there exist positive constants $k_{m}$ and $K_{m}$ such that
$$|\phi_{m}(t)|\leq K_{m}\|t\|^{m+|\lambda|-n}  \quad (0<\|t\|\leq k_m).$$
Accordingly, \be \label{estimate1}\Big|\int_{0}^{k_m}\cdots \int_{0}^{k_m}e^{-N\sum_{j=1}^{n}t_{j}}|\Delta(t)|^{\beta} \phi_{m}(t)\,d t_{1}\cdots d t_{n}\Big|<
\frac{C_{m}K_m}{N^{n_{\beta}+|\lambda|+m}},\ee
where  $$C_{m}= \int_{0}^{\iy}\cdots \int_{0}^{\iy}
  e^{-\sum_{j=1}^{n}t_{j}}|\Delta(t)|^{\beta}\|t\|^{|\lambda|+m-n}\,d t_{1}\cdots d t_{n}.$$

   For the contribution from the range $$(0,\iy)^{n}\backslash(0,k_m)^{n}=\bigcup (J_{1}\times \cdots \times J_{n}),$$
   where $J_j=(0,k_m)$ or $[k_m,\iy)$ but at least one of which is the infinite interval $[k_m,\iy)$. One knows there are ($2^{n}-1$) intervals in the union above.  Without loss of generality,  we only consider the case when  $J=(0,k_m)^{n-1}\times [k_m,\iy)$.

   Let $N_{0}$ be a value of $N$ for which the integral on the right-hand  side of \eqref{rightintegral} exists, and write
   $$\Phi_{m}(t)=\int_{0}^{t_1}\cdots \int_{0}^{t_{n-1}}\int_{k_m}^{t_n}
  e^{-N_{0}\sum_{j=1}^{n}v_{j}}|\Delta(t)|^{\beta} \phi_{m}(v)\,d v_{1}\cdots d v_{n},$$
   so that $\Phi_{m}(t)$ is continuous and bounded in $J=(0,k_m)^{n-1}\times [k_m,\iy)$. Let $L_m$ denote the supremum of $|\Phi_{m}(t)|$ in this range. When $N>N_0$, integration by parts produces
   \begin{align} I_{J}:&=\int_{J}
  e^{-N\sum_{j=1}^{n}t_{j}}|\Delta(t)|^{\beta} \phi_{m}(t)\,d t_{1}\cdots d t_{n}\nonumber\\
  &=\int_{J}e^{-(N-N_0)\sum_{j=1}^{n}t_{j}}
  e^{-N_{0}\sum_{j=1}^{n}t_{j}}|\Delta(t)|^{\beta} \phi_{m}(t)\,d t_{1}\cdots d t_{n}\nonumber\\
  &=\sum_{j=0}^{n-1}(N-N_0)^{j+1}\sum_{1\leq i_{1}<\cdots<i_{j}\leq n-1} \int_{(0,k_m)^{j}\times [k_m,\iy)}\Big[e^{-(N-N_0)\sum_{l=1}^{n}t_{l}}\Phi_{m}(t)\Big] d t_{i_{1}}\cdots d t_{i_{j}} d t_{n}
  ,\end{align}
   where $[f(t)]$ denotes the evaluation of $f(t)$ at $t_l=k_m$ if $l\neq i_{1},\ldots, i_{j}, n$.

   Thus
   \begin{align} |I_{J}|&\leq L_{m}\sum_{j=0}^{n-1}\binom{n-1}{j}   (N-N_0)^{j+1}    e^{-(n-j-1)(N-N_0)k_m} \int_{(0,k_m)^{j}\times [k_m,\iy)} e^{-(N-N_0)\sum_{l=n-j}^{n}t_{l}} d t_{i_{n-j}}\cdots  d t_{n}\nonumber\\
   &=L_{m}  e^{-(N-N_0)k_m}. \label{estimate2}
   \end{align}

Combining \eqref{estimate1} and \eqref{estimate2}, we immediately see that the integral on the right-hand side of \eqref{rightintegral} is $O(N^{-n_{\beta}-|\lambda|-m})$
 as $N\rightarrow \iy$, and the lemma is proved.
\end{proof}

\begin{coro} With the same assumptions as in the previous theorem and with $\lambda_i=\mu >0$ for all $i$,   we have as $N\to\infty$,
\begin{multline*} \int_{(0,\iy)^{n}}
  e^{-N\sum_{j=1}^{n}t_{j}}|\Delta(t)|^{\beta} q(t)\,d^{n}t=\prod_{j=0}^{n-1}
\frac{\Gamma(1+\beta/2+j\beta/2) \Gamma(\mu+j\beta/2)}
{\Gamma(1+\beta/2)}\frac{a_0}{N^{n_{\beta}+\mu n}}+O(N^{-n_{\beta}-\mu n-1})\end{multline*} where $$ a_0=\left[t^{-\lambda+1}q(t)\right]_{t=(0,\ldots,0)} .$$ \end{coro}

\begin{remark}The integrals $A_{j}$, given by \eqref{monomialintegral}, is in general very difficult to evaluate.
However, if $q(t)$ is a symmetric function and $\lambda_{1}=\cdots=\lambda_{n}$, then $A_{j}$ can be calculated by the Macdonald-Kadell-Selberg integrals,
see \cite{macdonald}, p. 386 (a) or the original papers \cite{kadell,kaneko}.
\end{remark}

\subsection{Laplace's method for contour integrals: a single saddle point}
\label{Laplacemethod}

Let $a$ and $b$ be two complex numbers. Let also $p(x)$ and $q(t)$ be analytic
functions of $x$ and $t=(t_{1},\ldots,t_{n})$ in the domains $\textbf{T}\subseteq \mathbb{C}$ and $\textbf{T}^{n}$, respectively. Finally, suppose that   $N$ is a positive parameter and $h(x)$ is a homogeneous analytic function of degree $\nu\geq 0$, i.e.,
$  h(cx)=c^{\nu}h(x), \forall\ \textrm{Re}\, c>0$.   In this subsection, we consider the asymptotic evaluation, as $N\to\infty$, of the following $n$-dimensional integral:
\begin{equation} \label{pqintegral}I_N=\int_{a}^b \cdots \int_{a}^ {b}
  \exp\{-N\sum\limits_{j=1}^{n}p(t_{j})\}h(\Delta(t))\, q(t)\,d t_{1}\cdots d t_{n}, \end{equation}
 in which  the integration between $a$ and $b$ is taken by following the contour $ {(a,b)_\mathscr{P}} $ in $\mathbb{C}$.

Like  the one-dimensional integral, to obtain the asymptotics of integrals one usually needs to deform the path through some special points at which $p'(x)=0$,
called \textit{saddle points}, we refer to sections 7 and 10, \cite{olver} for more details about \textit{saddle point} and paths of \textit{steepest descent}.

Recall that $x_{0}$ is a saddle point of order $\mu-1$ if
$$p'(x_0)=\cdots=p^{(\mu-1)}(x_0)=0,\  p^{(\mu)}(x_0)\neq 0,$$
where the integer $\mu\geq 2$. In particular, when $\mu=2$ it is called a simple saddle point. The most common cases for the integrals are that $p(x)$ has  (1) one simple
 saddle point; (2) two simple saddle points; (3) one saddle point of order 2, at interior points of the integration path.
Those occur  ubiquitously in Random Matrix Theory, corresponding to the hard-edge, bulk and soft-edge    limiting behavior.
Fortunately, deformations to the path through \textit{saddle points} for  our multi-dimensional integrals   can reduce to the one-dimensional case. But
great care must be taken in dealing with the part involving the Vandermonde determinant.

We will first consider  the case of a single saddle point of order $\mu-1$, which  generalizes both cases  (1) and (3).
We moreover assume the following conditions:
\begin{itemize}

\item[(i)]$p(x)$ and $q(t)$ are   single-valued and holomorphic in the domains $\textbf{T}\subseteq \mathbb{C}$ and $\textbf{T}^{n}$; $h(x)$ is a homogeneous analytic function of degree $\nu\geq 0$ in $\mathbb{C}$.

\item[(ii)] The integration path  $\mathscr{P}$  is independent of $N$. The endpoints $a$ and $b$ of  $\mathscr{P}$  are
finite or infinite, and $(a,b)_{\mathscr{P}}$ lies within $\textbf{T}$.

\item[(iii)] $p'(x)$ has exactly one   zero of order $\mu-1$ at an interior point $x_0$ of $\mathscr{P}$, i.e., $p'(x_0)=\cdots=p^{(\mu-1)}(x_0)=0, p^{(\mu)}(x_0)\neq0$;
in the neighborhoods of $x_0$ and $t_0=(x_0,\ldots,x_0)$, $p(x)$ and $q(t)$ can be expanded in convergent series of the form
$$p(x)=p(x_0)+\sum_{s=0}^{\iy}p_{s}(x-x_0)^{s+\mu}, q(t)=(t-t_0)^{\lambda-1} \sum_{j=0}^{\iy}q_{j}(t-t_0),$$
where $p_0\neq 0$ and $q_{j}(t)=\sum_{|k|=j}q_{j,k}\,t^{k}$.

\item[(iv)] There exists $N_{0}>0$ such that  $I_{N_{0}}$ converges absolutely at $(a,\ldots,a)$ and $(b,\ldots,b)$.

\item[(v)] $\textrm{Re}\{p(x)-p(x_0)\}$ is positive on $(a,b)_{\mathscr{P}}$, except at $x_0$, and is bounded away
from zero   as $x \rightarrow a\  \mathrm{or}\  b \ \textrm{along} \  \mathscr{P}$.
\end{itemize}

\begin{remark}[Steepest descent paths]
Let  $p_0:=(\mu!)^{-1} p^{(\mu)}(x_0)$.  Suppose that $x$ follows a path $\mathscr{P}$ such that the phase of $(x-x_0)^\mu p_0$ is exactly equal to zero.  Then the first non-null term in $\textrm{Re}\{p(x)-p(x_0)\}$  reaches its maximum value  while $\textrm{Im}\{p(x)-p(x_0)\}$ gets equal to zero.  This in turn implies that $-N\textrm{Re} \{p(x)-p(x_0)\}$ decreases at the maximum rate possible  as $x$ goes away from  the local maximum $x_0$, which justifies the fact that $\mathscr{P}$    is called the steepest descent path. For our special cases coming from Random Matrix Theory, we will always be able to deform a part of the contour of integration to one of steepest descent contours.  For the general case, however,  the condition $\textrm{Re}\{p(x)-p(x_0)\}>0$ is sufficient.
\end{remark}

\begin{remark}[Branch convention] In the rest of the article, some ambiguities could arise when dealing with fractional powers of complex numbers.  We will always use the following convention.  Let  $\textrm{ph}(z)$ denote the phase or argument of complex variable $z$. For an interior point $x_{0}$ of $(a,b)_{\mathscr{P}}$, we denote
 \be  \omega =  \textrm{angle of slope of}  \, \mathscr{P} \,  \textrm{at}
 \, x_{0} =  \lim \{\textrm{ph} (x - x_{0})\} \qquad ( x \rightarrow x_{0}\,  \textrm{along }  (x_{0},b)_{\mathscr{P}}).
 \ee
The value of $\omega_0 =\textrm{ph}(p_{0})$ is not necessary the principal one, but is chosen to satisfy
 \be |\omega_0+\mu\, \omega |\leq \frac{1}{2}\pi \; .\ee
The branch of $\textrm{ph}(p_{0})$ is used in constructing all fractional powers of $p_{0}$, so for example, $p^{1/\mu}_{0} $ means $\exp\{(\ln|p_{0}|+i\omega_{0})/\mu\}$.  When defining the variable  $w_{j}$ via
  $$w_{j}^\mu=\left(p(t_{j})-p(x_0)\right)/p_0,$$  the branch of   $\textrm{ph}(w_{j})$ is determined by
\be   \mu\,\textrm{ph}(w_{j}) \rightarrow \mu\, \omega  \qquad (\, w_{j}  \rightarrow x_{0}\,  \textrm{ along } \, (x_{0},b)_{\mathscr{P}}\,),\ee
and by continuity elsewhere.
\end{remark}

\begin{lemma}\label{lemmapaths} Assume conditions (i), (ii), and (iii) above. Let $\mathbf{D}\subset\mathbb{C}$ denote a small disk centered at $x_0$ and let  $w_{j}\,:\,\mathbf{D}\to\mathbb{C} $ be the map such that $w_{j}^\mu=(p(t_{j})-p(x_0))/p_0$.    Then, there exist paths $(\tau_1,x_0)_{\mathscr{P}}$ and $(x_0,\tau_2)_{\mathscr{P}}$  contained in $(a,b)_{\mathscr{P}}$ that can be respectively deformed into   paths $\mathscr{P}_1$ and $\mathscr{P}_2$, such that under the map $w_{j}$, $\mathscr{P}_1$ becomes a straight line path $L_{11}$ ending at the origin while $\mathscr{P}_2$  becomes a  straight line path $L_{21}$ starting at the origin.
\end{lemma}
\begin{proof}
For small $\|t-t_0\|$,   Condition (iii) and the Binomial theorem yield
$$w_j=(t_j-x_0)\{1+\frac{p_{1}}{\mu p_{0}}(t_j-x_0)+\cdots\}.$$
 Application of the inversion theorem for analytic functions
shows that for all sufficiently small $\rho>0$, the disk
$|t_j-x_0|<\rho$ is mapped conformally on a domain $\textbf{W}$ containing $w_j = 0$. Moreover,
if $w_j\in \textbf{W}$ then $t_j-x_0$ can be expanded in a convergent series
$$t_j-x_0=\sum_{s=1}^{\iy}c_{s}w_{j}^{s},$$
in which the coefficients $c_{s}$ are expressible in terms of the $p_{s}$. For example, $c_1=1$ and $c_2=-p_{1}/(\mu p_{0})$.

Let $\tau_{1}, \tau_{2}$ be points of $(a,x_0)_{\mathscr{P}}, (x_0,b)_{\mathscr{P}}$ respectively chosen sufficiently close to $x_0$ to ensure that the disk
$$|w_j|\leq\min\{|p(\tau_1)-p(x_0)|^{1/\mu}, |p(\tau_2)-p(x_0)|^{1/\mu}\} $$
is contained in $\textbf{W}$. Then $(\tau_1,x_0)_{\mathscr{P}}$ and $(x_0,\tau_2)_{\mathscr{P}}$ may be deformed to make their $w_j$ map straight lines $L_{11}$ and $L_{21}$, respectively. If
$$\kappa_j=(p(\tau_j)-p(x_0))/p_{0}, \quad j=1,2,$$
then $L_{11}$ and $L_{21}$ are directed line segments, respectively, from  $\kappa_1$ to $0$ and from $0$ to $\kappa_2$, as expected.
\end{proof}

\begin{theorem} \label{mainintegral1}With the foregoing assumptions (i)--(v), $\forall\, m\in \mathbb{N}$,
we have   as $N\rightarrow \iy$,
  \begin{multline} \int_{a}^b\cdots  \int_{a}^b
  \exp\{-N\sum\limits_{j=1}^{n}p(t_{j})\}h(\Delta(t))\, q(t)\, d t_{1}\cdots d t_{n}\\=e^{-nN p(x_{0})}
  \Big\{\sum_{j=0}^{m-1}\frac{A_{j}}{N^{(n_{\nu}+|\lambda|+j)/\mu}}+O(\frac{1}{N^{(n_{\nu}+|\lambda|+m)/\mu}})\Big\}\, .\nonumber\end{multline}
The coefficients $A_{j}$ of the asymptotic expansion are given by
 \be  A_j=\int_{L_{1}\cup L_{2}} \cdots \int_{L_{1}\cup L_{2}}
  \exp\{- p_{0}\sum\limits_{j=1}^{n} w_{j}^{\mu}\}h(\Delta(w))\,w^{\lambda-1} a_{j}(w)\,d w_{1}\cdots d w_{n},\ee
	where $L_1$ and $L_2$ denote semi-infinite straight line paths in the complex plane, the first ending with the line   $L_{11}$ and the  second starting with the line   $L_{21}$ (see Lemma \ref{lemmapaths}).
  \end{theorem}
\begin{proof}
We proceed in three steps.

Step 1. Let $L_{11}$ and $L_{21}$ be the line paths of Lemma \ref{lemmapaths}, respectively  going from  $  \kappa_1$ to $0$ and from 0 to $\kappa_2$.  Now let  $L_{1}=L_{12}\cup L_{11}$ and $L_{2}=L_{21}\cup L_{22}$ denote the semi-infinite lines, the first passing through $  \kappa_1$ and ending at $0$, the second starting at $0$ and passing through   $\kappa_2 $. Thus we have
\begin{multline} \label{leadingintegral} \int_{\tau_1}^{\tau_2}\cdots \int_{\tau_1}^{\tau_2}
  \exp\{-N\sum\limits_{j=1}^{n}p(t_{j})\}h(\Delta(t))\, q(t)\,d t_{1}\cdots d t_{n}=\\e^{-nN p(x_{0})}\int_{L_{11}\cup L_{21}} \cdots \int_{L_{11}\cup L_{21}}
  \exp\{-N p_{0}\sum\limits_{j=1}^{n} w_{j}^{\mu}\}\,h(\Delta(w)) f(w)\,d w_{1}\cdots d w_{n}, \end{multline}
where $$f(w)=q(t)\Big(\prod_{i<j}\sum_{s\geq 1}c_{s}\frac{w_{i}^{s}-w_{j}^{s}}{w_{i}-w_{j}}\Big)^{\nu}\prod_{j=1}^{n}\big(\sum_{s\geq 1}s\, c_{s}w_{j}^{s-1}\big).$$

For small $\|w\|$, $f(w)$ has a convergent expansion of the form
\be  f(w)=w^{\lambda-1} \sum_{j=0}^{\iy}a_{j}(w), \quad
a_{j}(w)=\sum_{|k|=j}a_{j,k}\,w^{k},\ee
in which the coefficients $a_{j,k}$ can be computed in terms of  $p_s$ and $q_{j,k}$. In particular, $a_0=q_{0}$.

Step 2. Following the   approach for the Watson's lemma, we define $f_{m}(w), m = 0,1,\ldots,$ by
\be  f(w)=w^{\lambda-1} \sum_{j=0}^{m-1}a_{j}(w)+w^{\lambda-1} f_{m}(w).\ee
Then $f_{m}(w)=O(\|w\|^{m})$. Set $n_{\nu}=\nu n(n-1)/2$ and
  \be \label{monomialintegral1}A_j=\int_{L_{1}\cup L_{2}} \cdots \int_{L_{1}\cup L_{2}}
  \exp\{- p_{0}\sum\limits_{j=1}^{n} w_{j}^{\mu}\}h(\Delta(w))\,w^{\lambda-1} a_{j}(w)\,d w_{1}\cdots d w_{n},\ee the integral on the right-hand side of \eqref{leadingintegral} is rearranged in the form
\begin{multline} \label{leadingintegral1} \int_{L_{11}\cup L_{21}} \cdots \int_{L_{11}\cup L_{21}}
  \exp\{-N p_{0}\sum\limits_{j=1}^{n} w_{j}^{\mu}\}h(\Delta(w))\, f(w)\,d w_{1}\cdots d w_{n}\\=\sum_{j=0}^{m-1}\frac{A_{j}}{N^{(n_{\nu}+|\lambda|+j)/\mu}}
  -\varepsilon_{m,1}(N)+\varepsilon_{m,2}(N). \end{multline}
Here \be \label{smallintegral1}\varepsilon_{m,1}(N)=\sum_{j=0}^{m-1}\sum_{J} \int_{J}
  \exp\{- N p_{0}\sum\limits_{j=1}^{n} w_{j}^{\mu}\}h(\Delta(w))\, w^{\lambda-1} a_{j}(w)\,d w_{1}\cdots d w_{n}, \ee
summed over $J=J_{1}\times \cdots \times J_{n}$, $J_{j}=L_{11}\cup L_{21} $ or $L_{12}\cup L_{22}$, but at least one of which is  $L_{12}\cup L_{22}$, and
\be \label{smallintegral2}\varepsilon_{m,2}(N)=\int_{L_{11}\cup L_{21}} \cdots \int_{L_{11}\cup L_{21}}
  \exp\{-N p_{0}\sum\limits_{j=1}^{n} w_{j}^{\mu}\}h(\Delta(w))\, w^{\lambda-1} f_{m}(w)\,d w_{1}\cdots d w_{n}. \ee

For $\varepsilon_{m,2}(N)$, splitting the integral domain into $2^{n}$ parts, for one of which, for instance,
the domain $L_{11}\times \cdots \times L_{11}\times L_{21}$,  substitute $w_{j}=\kappa_{1}v_{j}, j<n$ and $w_{n}=\kappa_{2}v_{n}$,  and note that Condition (v)
 implies
 \be\label{Rebound}\textrm{Re}\{p_0\kappa_{1}^{\mu}\}\geq \eta_{1}, \quad \textrm{Re}\{p_0\kappa_{2}^{\mu}\}\geq \eta_{1}\ee  for some positive $\eta_{1}$,  we have
\begin{align*} &\int_{L_{11}} \cdots \int_{L_{21}}
  \exp\{-N p_{0}\sum\limits_{j=1}^{n} w_{j}^{\mu}\}h(\Delta(w))\, w^{\lambda-1} f_{m}(w)\,d w_{1}\cdots d w_{n}\\
  &=   \int_{0}^{1} \cdots \int_{0}^{1}
  \exp\{-N p_{0}(\kappa_{1}^{\mu}v_{1}^{\mu}+\cdots+\kappa_{1}^{\mu}v_{n-1}^{\mu}+\kappa_{2}^{\mu}v_{n}^{\mu}\}h(\Delta(w))\, w^{\lambda-1} O(\|w\|^{m})
  \,d v_{1}\cdots d v_{n}\\
  &=O(N^{-(n_{\nu}+|\lambda|+m)/\mu}). \end{align*}

For $\varepsilon_{m,1}(N)$, without loss of generality, consider the case when $J_1=\cdots=J_{n-1}=L_{11}\cup L_{21}, J_{n}=L_{12}\cup L_{22}$.
Then \begin{align}I_{J}:&=  \int_{J}
  \exp\{- N p_{0}\sum\limits_{j=1}^{n} w_{j}^{\mu}\}h(\Delta(w))\,w^{\lambda-1}  a_{j}(w)\,d w_{1}\cdots d w_{n}\nonumber\\
  &= \int_{J}\exp\{- (N-N_{0}) p_{0}\sum\limits_{j=1}^{n} w_{j}^{\mu}\}
  \exp\{- N_{0} p_{0}\sum\limits_{j=1}^{n} w_{j}^{\mu}\}\,h(\Delta(w))\, w^{\lambda-1} a_{j}(w)\,d w_{1}\cdots d w_{n}.\end{align}
 Condition (v) implying that  $\textrm{Re}\{p_0 w_{j}^{\mu}\}\geq 0, 1\leq j<n$ for $w_{j} \in L_{11}\cup L_{21}$, combining with \eqref{Rebound} and Condition (iv)  we obtain
 \begin{align}|I_{J}|&\leq O(1) \, e^{-(N-N_{0})\eta_{1}}.\end{align}
%In the last equality we use the asymptotic expansion for incomplete Gamma Functions, see (2.02),  Chapter 4, \cite{olver}.

The combination of the results of this step with \eqref{leadingintegral} leads to
\begin{multline} \label{leadingintegral&small} \int_{\tau_1}^{\tau_2}\cdots \int_{\tau_1}^{\tau_2}
  \exp\{-N\sum\limits_{j=1}^{n}p(t_{j})\}h(\Delta(t))\, q(t)\,d t_{1}\cdots d t_{n}=\\e^{-nN p(x_{0})}
  \Big\{\sum_{j=0}^{m-1}\frac{A_{j}}{N^{(n_{\nu}+|\lambda|+j)/\mu}}+O(\frac{1}{N^{(n_{\nu}+|\lambda|+m)/\mu}})\Big\}, \end{multline}
as $N\rightarrow \iy$.

Step 3. It remains to consider the tail of the integral, that is, the contribution from
$(a,b)^{n}\setminus (\tau_{1},\tau_{2})^{n}=\bigcup_{J}J_{1}\times \cdots \times J_{n}$, $J_{j}=(\tau_{1},\tau_{2})$ or $(a,\tau_{1})\cup(\tau_{2},b)$ but at least one of which is $(a,\tau_{1})\cup(\tau_{2},b)$.
 Without loss of generality, we assume that $J_1=\cdots=J_{n-1}=(\tau_{1},\tau_{2})$ and $J_{n}=(a,\tau_{1})\cup(\tau_{2},b)$. From Condition (v) we know  that
\be  \textrm{Re}\{p(x)-p(x_0)\}\geq 0\,  ( x\in (a,b)_{\mathscr{P}}) \quad \text{and} \quad \textrm{Re}\{p(x)-p(x_0)\}\geq \eta_2\, ( x\in (a,\tau_{1})_{\mathscr{P}}\cup(\tau_{2},b)_{\mathscr{P}})\ee
for some   $\eta_2>0$. Therefore, $$N\textrm{Re}\{p(t_{n})-p(x_0)\}\geq (N-N_{0})\eta_{2}+N_{0}\textrm{Re}\{p(t_{n})-p(x_0)\},$$
and condition (iv) shows
\begin{multline} \label{tailintegral} \Big|  \int_{J}
  \exp\{-N\sum\limits_{j=1}^{n}(p(t_{j})-p(x_0))\}h(\Delta(t))\, q(t)\,d t_{1}\cdots d t_{n}\Big|\leq\\e^{-(N-N_{0})\eta_{2}}\big|e^{-nN_0 p(x_{0})}\big|
    \int_{J}
 \Big |\exp\{-N_0\sum\limits_{j=1}^{n} p(t_{j})\}h(\Delta(t))\, q(t)\,d t_{1}\cdots d t_{n}\Big|. \end{multline}
Thus the asymptotic expansion  \eqref{leadingintegral&small} is unaffected by  the tail integrals and the theorem follows.
\end{proof}

%%\begin{remark}The coefficients $A_{j}$, given by \eqref{monomialintegral1}, in general is very difficult to evaluate.
%%However, in the present paper special cases $\mu=2,3$ and the first coefficient $A_{0}$ are our main concern.  Under the circumstances, $A_{0}$ ($\mu=2$) can be evaluated by  using Selberg's integral formula.
%%\end{remark}

\begin{remark}\label{remarkmu} If we focus on the leading term with  the coefficient $A_0$, then there is an immediate generalization of   the previous theorem which will be used later in the article.  Let $g(t)$ be analytic in the whole $\mathbb{C}^n$.  Now in the integral \eqref{pqintegral}, replace $q(t)$ by $q(t)g(N^{1/\bar{\mu}}(t-t_{0}))$, and assume $\bar \mu\geq \mu$.
  %%where $\mu-1$ is equal to the order of the saddle point $x_0$ of $p(x)$.
  At first sight, the factor $g(N^{1/\bar\mu}(t-t_{0}))$ may become very large as $N\to \infty$.  However, by repeating the analysis done in steps 2 and 3 (and paying a particular attention to the rescaling of the variables), we see that $g(N^{1/\bar\mu}(t-t_{0}))$ does not lead to a new significant contribution to the leading coefficient.  Most importantly, if $\bar \mu=\mu$, then we get a very simple formula:
\begin{multline*} \int_{a}^{b}\cdots \int_{a}^{b}
 \exp\{-N\sum\limits_{j=1}^{n}p(t_{j})\}h(\Delta(t))\, q(t)\,g(N^{1/\mu} (t-t_{0}))\, d t_{1}\cdots d t_{n}\\
        =e^{-nN p(x_{0})}
 \Big\{
  \frac{A'_{0}}{N^{(n_{\nu}+|\lambda|)/\mu}}+O(N^{-(n_{\nu}+|\lambda|+1)/\mu})
 \Big\},
 \end{multline*}
  where
        \be \label{monomialintegral1p}A'_0=q_{0}\int_{L_{1}\cup L_{2}} \cdots \int_{L_{1}\cup L_{2}}
  \exp\{- p_{0}\sum\limits_{j=1}^{n} w_{j}^{\mu}\}h(\Delta(w))\, w^{\lambda-1} \,g(w)\,d w_{1}\cdots d w_{n}.\ee
\end{remark}

\subsection{Laplace's method for contour integrals: two simple saddle points}\label{section2saddle}

Once again we consider the integral \eqref{pqintegral}.  This time however, we suppose that $p(x)$ has two simple saddle points $x_\pm$, which means $$p'(x_\pm)=0\qquad  \text{and}\qquad  p_\pm:=p''(x_\pm) \neq 0.$$
This case is quite different from the one-dimensional case because of the Vandermonde determinant.  Our assumptions are:
\begin{itemize}
\item[(i)] $p(x)$ and $q(t)$ are   single-valued and holomorphic in the domains $\textbf{T}\subseteq \mathbb{C}$ and $\textbf{T}^{n}$;
$h(x)$ is a homogeneous analytic function of degree $\nu\geq 0$ in $\mathbb{C}$; for simplicity,  we also suppose that $q(t)$ is symmetric and $h(x)=h(-x)$.

\item[(ii)] The integration path  $\mathscr{P}$  is independent of $N$. The endpoints $a$ and $b$ of  $\mathscr{P}$  are
finite or infinite, and $(a,b)_{\mathscr{P}}$ lies within $\textbf{T}$.

\item[(iii)] $p'(x)$ has exactly two simple  zeros  at   interior points $x_{+},x_{-}$ of $\mathscr{P}$; setting $t_{j,+}=(x_{+},\ldots,x_{+},x_{-},\ldots,x_{-})$
consisting  of $j$ $x_{+}$'s and $(n-j)\, x_{-}$'s, $p_{\pm}=p''(x_\pm)$ and $q_{j,+}=q(t_{j,+})$.

\item[(iv)] There exists $N_{0}>0$ such that  $I_{N_{0}}$ converges absolutely at $(a,\ldots,a)$ and $(b,\ldots,b)$.

\item[(v)] $\textrm{Re}\{p(x)-p(x_{\mp})\}$ is positive on $(a,x_0]_{\mathscr{P}}$ and $[x_0,b)_{\mathscr{P}}$ respectively for some $x_0\in (a,b)_{\mathscr{P}}$, except at $x_{-},x_{+}$, and is bounded away
from zero   as $x \rightarrow a$, or $x \rightarrow b$ along $\mathscr{P}$.
\end{itemize}

\begin{lemma}\label{lemmapaths2} Assume conditions (i), (ii), and (iii). Let $\mathbf{D}_{\pm}\subset\mathbb{C}$ denote small disks centered at $x_\pm$ and let  $w\,:\,\mathbf{D}\to\mathbb{C} $ be the map such that $w(x)^\mu=(p(x)-p(x_0))/p_0$.    Then, there exist paths $(\tau_1,x_\pm)_{\mathscr{P}}$ and $(x_\pm,\tau_2)_{\mathscr{P}}$  contained in $(a,b)_{\mathscr{P}}$ that can be respectively deformed into   paths $\mathscr{P}^\pm_1$ and $\mathscr{P}^\pm_2$, such that under the map $w$, $\mathscr{P}^\pm_1$ becomes a straight line path $L^\pm_{11}$ ending at the origin while $\mathscr{P}^\pm_2$  becomes a  straight line path $L^\pm_{21}$ starting at the origin.
\end{lemma}
\begin{proof}It suffices to subdivide the path $(a,b)_{\mathscr{P}}$ into two disjoint sub-paths $(a,b)^\pm_{\mathscr{P}}\ni x_\pm$ and then proceed as for Lemma \ref{lemmapaths}.\end{proof}

\begin{theorem}\label{teo2saddle}With the foregoing assumptions (i)--(v),   as $N\rightarrow \iy$,
  \begin{multline} \label{mainintegral2} \int_{a}^{b}\cdots \int_{a}^{b}
  \exp\{-N\sum\limits_{j=1}^{n}p(t_{j})\}h(\Delta(t))\, q(t)\,d t_{1}\cdots d t_{n}=\\
        \sum_{l=0}^{n}\binom{n}{l}\exp\{-N \big(lp(x_{+})+(n-l)p(x_{-})\big)\}
  \,\Big(q_{l,+}\,B_{l,+}N^{-(l_{\nu}+(n-l)_{\nu}+n)/2}\big(1+O(N^{-1/2})\big)\Big) \, . \end{multline}
Here the coefficients $B_{l,+}$ are given by
\be \label{monomialintegral2}B_{j,+}= \int_{(L_{1,+}\cup L_{2,+})^{j}\times (L_{1,-}\cup L_{2,-})^{n-j}}
  \exp\Big\{- \frac{p_{+}}{2}\sum\limits_{l=1}^{j} w_{l}^{2}-\frac{p_{-}}{2}\sum\limits_{l=j+1}^{n} w_{l}^{2}\Big\}\,h(\Delta_{j,+}(w))\,d w_{1}\cdots d w_{n}\, ,\ee
	where $L_{1,\pm}$ and $L_{2,\pm}$ denote semi-infinite straight lines such that $L_{1,\pm}\supset  L^\pm_{11}$ and $L_{2,\pm} \supset L^\pm_{21}$ respectively end and start at the origin (see Lemma \ref{lemmapaths2} above). \end{theorem}
\begin{proof}

As in step 1 of Section \ref{Laplacemethod},  let $L_{1,+}\cup L_{2,+}$ and $L_{1,-}\cup L_{2,-}$ denote  the unions of two half-lines
respectively corresponding to $x_{+}$ and $x_{-}$.  Thus,  the $n$-dimensional integral can now be broken into a sum of $2^n$ terms, each being an $n$-dimensional integral in which each variable is in the neighborhood of either $x_+$ or $x_-$. When $j$ variables are in the neighborhood of $x_+$, the Vandermonde determinant becomes
$$ \Delta_{j,+}(w)=(x_{+}-x_{-})^{j(n-j)}\prod_{1\leq p<q\leq j}(w_{p}-w_{q})\prod_{j< p<q\leq n}(w_{p}-w_{q}).$$
By the symmetry assumption, there are $\tbinom{n}{j}$ terms amongst the $2^n$ terms that produce such a Vandermonde determinant. By proceeding as for Theorem \ref{mainintegral1} and focusing only on the dominant contribution one rapidly establishes the desired asymptotic formula.
\end{proof}

Notice the fact that $$l_{\nu}+(n-l)_{\nu}
%%=\frac{\nu}{2} \Big(l(l-1)+(n-l)(n-l-1)\big)
=\nu(l-\frac{n}{2})^{2}+\nu  n(n-2)/4 \quad (l=0,1, \ldots, n)$$
attains its minimum value at $l=m$ if $n=2m$ or $l=m,m-1$ if $n=2m-1$, which shows that probably only one or two terms in the sum of  \eqref{mainintegral2} give
 a major contribution.

\subsection{Integrals of Selberg type}\label{sectionabsval}

Recall that our aim is to evaluate the  integrals such as \eqref{genint} when $N\to \infty$.   In Theorems \ref{mainintegral1} and \ref{teo2saddle}
$h(\Delta(t))$ must be homogeneous and analytic.  The latter condition is problematic since we want to calculate integrals involving
$|\Delta(t)|^\nu$.  Of course, when $t\in \mathbb{R}^n$ and $\nu$ is even, then we can
set $h(\Delta(t))=(\Delta(t))^\nu$.  For $\nu$ not even, it will be enough to use the following integral representation:\footnote{Another method for treating the absolute value of integral \eqref{genint}  was proposed in \cite{df}.  Its first step consists in changing  the $n$ contours $\mathcal{C}$ into a series of $n$ distinct contours,  $\mathcal{C}_1$ for variable $t_1$,  $\mathcal{C}_2$ for variable $t_2$, and so on, so that $|\Delta(t)|^\nu$ can be written as $\Delta(t)^\nu$ without affecting the value of the integral.  For instance, if $\mathcal{C}=\mathbb{R}$, then the contours $\mathcal{C}_1,\ldots, \mathcal{C}_n$ can be chosen such that they all start at $-\infty$  and comply with $\Re ({t_n})\leq \ldots \leq \Re (t_1)\leq \infty$.   Although with this method, one easily predicts the correct asymptotic expansions, we prefer not to use it because the nature of the contours  $\mathcal{C}_j$ depends on the specific integrals that must be evaluated, and because treating rigorously a succession of $n$ parametric integrals is much harder than what we use here.     }
\begin{equation}%\label{repabs}
|x|^{\nu}=c_\nu \int_0^{\infty}r^{-\nu-1}H_\nu(r x) d r, \label{integralrepofabsolutevaluepowerfunction1}
\end{equation}
which holds if  $x\in \mathbb{R}$, $\nu>-1$ and $\nu \neq 0, 2, 4, \ldots$. Here
\be \label{defH} H_{\nu}(r)=\sum_{j=0}^{[\nu/2]}\frac{1}{(2j)!}(-r^{2})^j -\cos(r)\qquad \text{and}\qquad c_\nu =\frac{2}{\pi}\sin\left(\frac{\pi\nu}{2}\right)\Gamma(\nu+1).\ee
This can be computed by complex analysis method when $-1<\nu<0$, see  Gradshteyn and Ryzhik's book \cite{grad}; otherwise, first integrate by parts and use the former result.

Let us now consider
\begin{equation} \label{pqintegralabs}I_{N}=\int_{\mathscr{P}}\cdots \int_{\mathscr{P}}
  \exp\{-N\sum\limits_{j=1}^{n}p(t_{j})\}|\Delta(t)|^\nu\, q(t)\,d t_{1}\cdots d t_{n}, \end{equation}
where the functions $p$ and $q$ are as  mentioned  previously while   $\mathscr{P}$ denotes some interval  $(a,b)\subseteq \mathbb{R}$ or one circle in the complex plane.  Assume that $\nu$ is not even and that the path  $\mathscr{P}=(a,b)\subseteq \mathbb{R}$, then
 $$  I_{N}=c_\nu \int_{(\!a,b)^{n}}
  \exp\{-N\sum\limits_{j=1}^{n}p(t_{j})\}\, q(t)\,\bigg(\int_{0}^{\infty}r^{-\nu-1}H_\nu (r \Delta(t))\, d r\bigg) \,d^{n}t.
$$
%%Given that $H_\nu (r\Delta(t))$ is analytic in $\Omega=\mathbb{C}^n\times \mathbb{R}_+$ and bounded in any compact subset of $\Omega$,
By Fubini's theorem, we   rewrite
\begin{equation}  I_{N}=c_\nu \int_{0}^{\infty}r^{-\nu -1}\bigg(\int_{(\!a,b)^{n}}
  \exp\{-N\sum\limits_{j=1}^{n}p(t_{j})\}\, q(t)\,H_\nu (r\Delta(t))\, d^{n} t\bigg) d r. \label{lineintabs}\end{equation}
Likewise, if $\nu$ is not even and $\mathscr{P}$ is the unit circle $\mathbb{T}:=\{z:|z|=1\}$, setting $t_j=e^{i\theta_j}$, since
$$|\Delta(t)|=|\prod_{1\leq j<k\leq n}(2 \sin\frac{\theta_j-\theta_k}{2})|\ \  \mbox{and} \ \  \prod_{1\leq j<k\leq n}(2 \sin\frac{\theta_j-\theta_k}{2})=\Delta(t)\prod_{j=1}^{n}(it_j)^{-(n-1)/2},$$
we have
\begin{equation}  I_{N}=c_\nu \int_{0}^{\infty}r^{-\nu-1} \bigg(\int_{\mathbb{T}^{n}}
  \exp\{-N\sum\limits_{j=1}^{n}p(t_{j})\}\, q(t)\, H_\nu \big(r \Delta(t)\prod_{j=1}^{n}(it_j)^{-(n-1)/2}\big)\,d^{n} t  \bigg)  d r.
\label{circleintabs}\end{equation}

Notice that the precise form of the integral over $r$ in \eqref{lineintabs} and \eqref{circleintabs} allows a rescaling of the variable, so the function $H_{\nu}(r)$ plays a similar role like  a homogeneous function of degree $\nu$, and we can apply established theorems for the interior $n$-dimensional integrals on the right-hand side of \eqref{lineintabs} and \eqref{circleintabs}. Finally, since  we are concerned on the leading term as $N \rightarrow \infty$ in this article,  assuming that we can deform the line segments in \eqref{monomialintegral1} or in \eqref{monomialintegral2} to the real line, then we can once more interchange the order of integration and reconstruct the absolute value by integrating over $r$.

  %%According to \eqref{defH}, the function $H_\nu (r\Delta(t))$ can be decomposed as $\sum_{j\geq [\nu/2]}(-1)^{j+1}(r\Delta(t))^{2j}/(2j)!$, which is uniformly convergent in any compact subset of $\Omega$.  We can then apply our previous theorems for each term $(r\Delta(t))^{2j}$.  At first sight, each term seems to lead to a different asymptotic development in $N^{-1}$, but the precise form of integral over $y$ allows a rescaling of the variables that makes every single term in the series of $H_\nu (r\Delta(t))$ contribute asymptotically as a homogeneous function of $\Delta(t)$ of degree $\nu$.  Finally, assuming that we can make the line segments in \eqref{monomialintegral1} or in \eqref{monomialintegral1} arbitrarily close to the real line without affecting the absolute convergence of the integrals, then we can once more interchange the order of integration and reconstruct the absolute value by integrating over $r$.

In what follows, we say that the integral $I_N$ as above satisfies the condition (vi):  after an appropriate change of variable, the factor $\xi$ in $H_\nu (r\xi)$, depending on the  variables $w$ and the saddle points, becomes a real-valued variable (the line segments $L_i$, coming from the integration in the neighborhood of the saddle points,  should first  be deformed to the real line).  All the examples  considered in the article satisfy this condition.  The next corollaries immediately follow from Theorem \ref{mainintegral1}, Remark \ref{remarkmu}, Theorem \ref{teo2saddle}, and the use of the integral representation \eqref{integralrepofabsolutevaluepowerfunction1}. Similar results hold for the integral \eqref{circleintabs}.

\begin{coro}\label{onesaddle} Under the foregoing  assumptions (i)-(v) of Section \ref{Laplacemethod} and (vi) above,  let
 \be\nonumber I_{N,n}=\int_{(\!a,b)^{n}}
  \exp \{-N\sum\limits_{j=1}^{n}p(t_{j}) \}|\Delta(t)|^\nu\, q(t)\,g(N^{1/\mu}(t-t_0))\,d^{n} t
        \ee
        where $g(t)$ is analytic in $\mathbb{C}^n$ and  $p(x)$ admits one  saddle point $x_0$ of order $\mu-1$.
 Then, as $N\to\infty$,
\begin{equation*} I_{N,n}\sim \frac{e^{-nNp(x_0)}}{N^{(n_\nu+n)/\mu}}  A_0\,q(x_0,\ldots,x_0)  \end{equation*}
where
\begin{equation*} A_0=\int_{\mathbb{R}^{n}} \exp\Big\{-p^{(\mu)}(x_0)/\mu!\sum_{j=1}^n w_j^\mu\Big\}\, g(w)\,|\Delta(w)|^\nu\, d^{n}w.
  \end{equation*}
 \end{coro}

\begin{coro}\label{twosaddle} Under the foregoing assumptions (i)-(v) of Section \ref{section2saddle} and (vi) above,  let
 \be\nonumber I_{N,n}=\int_{(\!a,b)^{n}}
  \exp\{-N\sum\limits_{j=1}^{n}p(t_{j})\}|\Delta(t)|^\nu\, q(t)\,d^{n} t
        \ee
        where $p(x)$ admits two simple saddle points $x_+, x_-$, and ${\rm{Re}}\{x_+-x_-\}\geq 0$.  Moreover, let $p_\pm=p''(x_\pm)$ and $\Gamma_{\nu,m}$ be given in \eqref{constgamma}.  If ${\rm{Re}}\{{p(x_+)}\}={\rm{Re}}\{{p(x_-)}\}$, then  as $N\to\infty$,
$$I_{N,2m}\sim\tbinom{2m}{m}(\Gamma_{\nu,m})^2\frac{(x_+-x_-)^{\nu m^2}}{(\sqrt{p_+p_-})^{m+\nu m(m-1)/2}}\frac{e^{-mN(p(x_+)+p(x_-))}}{N^{m+\nu m(m-1)/2}}q(x_+^{m},x_-^{m})$$
while
\begin{multline*} I_{N,2m-1}\sim \tbinom{2m-1}{m}\Gamma_{\nu,m-1}\Gamma_{\nu,m}
\frac{(x_+-x_-)^{\nu m(m-1)}}{(\sqrt{p_+p_-})^{m+\nu m(m-1)/2}}\frac{e^{-m N(p(x_+)+p(x_-))}}{N^{(2m-1 + \nu (m-1)^{2})/2}}\\
\times \Big(e^{Np(x_+)}(\sqrt{p_+})^{1+\nu(m-1)}q(x_+^{m-1},x_-^{m})+e^{Np(x_-)}(\sqrt{p_-})^{1+\nu(m-1)}q(x_+^{m},x_-^{m-1})\Big).  \end{multline*}
 \end{coro}

We conclude this section by applying the two last corollaries to  the study of the asymptotic behavior of the generalized Airy function \cite{des}:
\be \mathrm{Ai}^{(\alpha)}(s)=\frac{1}{(2\pi)^n}\int_{\mathbb{R}^n}e^{ip_3(t)/3}|\Delta(t)|^{2/\alpha}{\phantom{j}}_0\mathcal{F}_0^{(\alpha)}(s;i t) \,d t\ee
where it is assumed that the variables $s=(s_1,\ldots,s_n)\in \mathbb{R}^n$.  Note that the generalized  Airy function  obviously reduces to the classical Airy function when $n=1$.     Moreover, for $\alpha=2$, the above $n$-dimensional integral is proportional to Kontsevich's matrix Airy function $A(S)$, where $S$ denotes a $n\times n$ hermitian matrix with eigenvalues $s_1,\ldots,s_n$ (see Section 4 in \cite{kon}).  It is also worth noting that in the case where $\alpha=2$ and $s_1=\cdots=s_n=u$, then the above integral representation can be reduced to a very simple determinantal  formula (see for insante Section 4 in \cite{df} or Ref.\ 8 therein):
\be \mathrm{Ai}^{(2)}(u,\ldots,u)=(-1)^{n(n-1)/2}n!\det\left[\frac{d^{i+j-2}}{du^{i+j-2}}\mathrm{Ai}(u)\right]_{i,j=1}^n\,. \ee

Before displaying the asymptotics of  the generalized Airy function, we recall the following shorthand notation:
 when $A, B\in \mathbb{R}$,    $(A+Bs)$ stands for $(A+Bs_1,\ldots, A+B s_n)$.

\begin{proposition}\label{propAiry}Let $x$ be a real positive variable.  Then as $x\to \infty$,
\begin{equation}
\mathrm{Ai}^{(\alpha)}(x+x^{-1/2}s)\sim\frac{\Gamma_{2/\alpha,n}}{(2\pi)^{n}2^{(n+n(n-1)/\alpha)/2}}
\frac{e^{-\frac{2n}{3}x^{3/2}}e^{-p_1(s)} }{x^{(n+n(n-1)/\alpha)/4}},  \label{airy1}\end{equation}
and for $n=2m$
\begin{equation}
\mathrm{Ai}^{(\alpha)}(-x+x^{-1/2}s)\sim\tbinom{2m}{m}\frac{(\Gamma_{2/\alpha,m})^2}{(2\pi)^{n}}
(2\sqrt{x})^{-m+m(m+1)/\alpha} \,e^{-ip_1(s)} \!\!{\phantom{j}}_1F_1^{(\alpha)}(m/\alpha;n/\alpha;2is)\label{airy2}.  \end{equation}
\end{proposition}
\begin{proof} \eqref{airy1} and \eqref{airy2} originate from integrals evaluated around  one  simple saddle point and two simple saddle points, respectively.  For \eqref{airy1},  set $N=x^{3/2}$.  Simple manipulations and the use of  Proposition \ref{proppractical} lead to
\be \mathrm{Ai}^{(\alpha)}(N^{2/3}+N^{-1/3}s)=\frac{N^{(n+n(n-1)/\alpha)/3}}{(2\pi)^n}\int_{\mathbb{R}^n}e^{N(ip_3(t)/3+ip_1(
t))} |\Delta(t)|^{2/\alpha} \!\!{\phantom{j}}_0\mathcal{F}_0^{(\alpha)}(s;it) dt.\nonumber \ee
We thus have an integral like in Corollary \ref{onesaddle} with $p(t_j)=-it_j^3/3-it_j$, $q(t)=\!\!\!{\phantom{j}}_0\mathcal{F}_0^{(\alpha)}(s;it)$, and $g(t)=1$.  The function $p$ has two simple saddle points at $\pm i$.  With $x_0=i$, we have $p_0=p''(x_0)/2=1$ which implies that the steepest descent path near $x_0$ would follow the horizonal  line, as desired.  We may thus apply  Corollary \ref{onesaddle} to the case $\mu=2$, $x_0=i$, $p_0=1$, $p(x_0)=2/3$, $q(x_0,\ldots,x_0)=e^{-p_1(s)} $  and \eqref{airy1} follows immediately.

For \eqref{airy2}, we also let $N={x}^{3/2}$.  In the definition of $\mathrm{Ai}^{(\alpha)}(s)$, substitute $t_j$ by  $N^{1/3}t_j$ and apply Proposition \ref{proppractical}, which yields
\be \mathrm{Ai}^{(\alpha)}(-N^{2/3}+N^{-1/3}s)=\frac{N^{(n+n(n-1)/\alpha)/3}}{(2\pi)^n}\int_{\mathbb{R}^n}\prod_{j=1}^n e^{-N\, p(t_j)}|\Delta(t)|^{2/\alpha}
 \!\!{\phantom{j}}_0\mathcal{F}_0^{(\alpha)}(is;t) dt
\nonumber\ee
where
$p(t_j)=i\left(-t_j^3/3+t_j\right)$.
This function has 2 simple saddle points, namely $x_\pm=\pm 1$. This time we have to consider both of them because they are already on the path of integration.  We have $p(x_\pm)=\pm 2i/3$, $p_\pm=p''(x_\pm)=\mp 2i$.  This means that the steepest descent path is  given by
\be \mathscr{P}=\left\{-1+\tau e^{-i\pi/4}\,:\,\tau \in(-\infty,\sqrt{2}]\right\}\cup \left\{1+\tau e^{i\pi/4}\,:\, \tau\in[-\sqrt{2},\infty)\right\}.\nonumber\ee
By making the change of variables $w_j\mapsto e^{\pm i\pi/4}v_j$  (for saddle points $x_\pm$) in \eqref{monomialintegral2}, we see that the variables $v_j$ follow the real line,  so the assumption (vi) is fulfilled and \eqref{airy2} follows from the previous corollary.
\end{proof}

\begin{remark}\label{airytosine} The comparison of \eqref{eqlimitingbulk} in Theorem \ref{theobulk} and \eqref{airy2} in Proposition \ref{propAiry} shows that the bulk limit can be treated as a rescaled limit of the generalized Airy function. This relation is analogous to that between bulk and soft-edge limiting point processes for $\beta$-ensembles (see e.g.,  Corollary 3 of \cite{vv}).
\end{remark}

\section{Scaling limits}
\label{rmt}
In this section, we prove the theorems and corollaries given in the introduction.  Most of the proofs rely on Corollaries \ref{onesaddle} and \ref{twosaddle} or similar results for integrals on the torus $\mathbb{T}^{n}$, more precisely for \eqref{circleintabs}.  All the integrals considered here fulfill the assumptions (i) to (vi) given in the last section.  Note that when we talk about deforming the contours of integration, we implicitly suppose that either the power $\nu$ of the absolute value of Vandermonde determinant is even and the variables are real, or the integral representation \eqref{integralrepofabsolutevaluepowerfunction1} is being used.  Also in this section is the assumption that a deformation of contours is made as long as the integrand is analytic and absolutely integrable over the concerned region of the complex plane.

We will frequently use the symbol $\mathscr{L}_{a}^b$ to denote a straight line path from $a$ to $b$.  Additionally, $\mathscr{M}_a^b$ will denote a semi-circular path in the positive direction, starting at $a$ and ending at $b$ with radius $|a-b|/2$.  Dyson index $\beta$ and its duality   $\beta':=4/\beta$  will be used alternately.

Before starting our computation we first review   the integral representation of  hypergeometric functions $\!\!\!{\phantom{j}}_2 F_1^{(\alpha)}(a,b;c;s)$ and $\!\!\!{\phantom{j}}_1 F_1^{(\alpha)}(a;c;s)$, due to Yan \cite{yan} and Forrester \cite{forrester0}, especially \cite{forrester} for more details.  For $\alpha>0$, $\textrm{Re}\{\nu_1\}>-1$ and $ \textrm{Re}\{\nu_2\}>-1$,
\begin{equation} \label{yan21}\!\!\!{\phantom{j}}_2F_{1}^{(\alpha)}(a,b;c;s_1,\ldots,s_n)=\frac{1}{S_n(\nu_1,\nu_2,1/\alpha)}
\int_{[0,1]^{n}} {\!\!\!{\phantom{j}}_1\mathcal{F}_{0}^{(\alpha)}(a;s;t)}\, D_{\nu_1,\nu_2,1/\alpha}(t)\,d^{n}t,
\end{equation}
where $\nu_1=b- (n-1)/\alpha -1$, $\nu_2=c-b- (n-1)/\alpha-1$  and
\begin{equation}D_{\nu_1,\nu_2,1/\alpha}(t)=\prod_{i=1}^nt_i^{\nu_1}(1-t_i)^{\nu_2} \prod_{1\leq i < j \leq n}|t_i-t_j|^{2/\alpha}.
\end{equation}

If $\textrm{Re}\{\nu_1\}>-1$, the right-hand integral of \eqref{yan21} can be analytically continued so that it is valid for  $\textrm{Re}\{\nu_2\}\leq -1$ but  by replacing the interval $[0,1]$ with the counter-clockwise circle $\mathbb{T}$, especially in the case of interest ($a=-N$) \cite{forrester0,bf}:
\begin{align} \label{forrester21}\!\!\!{\phantom{j}}_2F_{1}^{(\alpha)}(a,b;c;s)&=\frac{e^{i\pi (b-c)n}}{M_n(b-c,c+1+(n-1)/\alpha,1/\alpha)}\nonumber\\
&\times \frac{1}{(2\pi i)^{n}}
\int_{\mathbb{T}^{n}} {\!\!\!{\phantom{j}}_1\mathcal{F}_{0}^{(\alpha)}(a;s;1-t)}\, D_{\nu_2+(n-1)/\alpha,\nu_1,1/\alpha}(t)\,d^{n}t.
\end{align}
Note that the constant $M_n(a,b,\alpha) $ is given in the appendix.  Likewise, for $\textrm{Re}\{c-a\}>-1$, we have
\begin{align} \label{forrester11}\!\!\!{\phantom{j}}_1F_{1}^{(\alpha)}(a;c+1+(n-1)/\alpha;s)&=\frac{e^{-i\pi n a}}{M_n(-a,c,1/\alpha)} \frac{1}{(2\pi i)^{n}}
\int_{\mathbb{T}^{n}} \prod_{j=1}^n t_j^{a-1}(1-t_j)^{c-a} |\Delta(t)|^{2/\alpha} {\!\!\!{\phantom{j}}_0\mathcal{F}_{0}^{(\alpha)}(s;t)}\,d^{n}t.
\end{align}

\subsection{Hermite $\beta$-ensemble}

 It follows from the duality relation in Proposition 7 \cite{des} that
\begin{align*} K_N(s_1,\ldots,s_n):&=\Big\langle \prod_{i=1}^N\prod_{j=1}^n(x_i-s_j)\Big\rangle_{\HBE}\\
&=(-i)^{nN}2^{\beta'\! n(n-1)/4+n/2}(\Gamma_{\beta'\!,n})^{-1}e^{p_2(s)}\int_{\mathbb{R}^n}\prod_{j=1}^n t_j^N e^{-t_j^2}|\Delta(t)|^{\beta'\!}     {\!\!\!{\phantom{j}}_{0}\mathcal{F}_0^{(2/\beta'\!)}}(-2is;t)d^{n}t.
\end{align*}
Set $$s_{j}\mapsto \sqrt{2N}(u+\frac{s_j}{\rho N}) \qquad \mbox{and} \qquad t_{j}\mapsto \sqrt{2N}t_j.$$
By using Proposition \ref{proppractical}, the weighted quantity
$$\varphi_{N}(s_1,\ldots,s_n):=e^{-\frac{1}{2}p_2(s)}K_N(s_1,\ldots,s_n)$$
  becomes  \be \varphi_{N-l}\big(\sqrt{2N}(u+\frac{s}{\rho N})\big)=(-i\sqrt{2N})^{n(N-l)}(2\sqrt{N})^{\beta^{'}\!n(n-1)/2+n}(\Gamma_{\beta^{'}\!,n})^{-1}e^{nN u^{2}+p_2(s)/(\rho^{2}N)}I_{N,n},\label{chHE1}\ee
where \begin{align} I_{N,n}&=\int_{\mathbb{R}^n} \exp\{-N\sum\limits_{j=1}^n p(t_j)\}|\Delta(t)|^{\beta'\!}  q(t) d^{n}t\nonumber\\
&=
c_{\beta'\!} \int_{0}^{\infty}r^{-\beta'\! -1}\bigg(\int_{\mathbb{R}^n}
  \exp\{-N\sum\limits_{j=1}^{n}p(t_{j})\}\, q(t)\,H_{\beta'\!} (r\Delta(t))\, d^{n} t\bigg) d r,\label{chHE2}\end{align}
  and
\be p(x)=2x^{2}+4iu x-\ln x, \qquad q(t)=\prod_{j=1}^n t_j^{-l} {\!\!{\phantom{j}}_{0}\mathcal{F}_0^{(2/\beta'\!)}}(-4is/\rho;t+iu/2)\label{chHE3}.\ee
Here $l$ is a fixed integer, and it is assumed that  $l=0$ for $n=2m$ while $l=0,1$ for  $n=2m-1$.

Since $p'(x)=4x+4iu -\frac{1}{x}$, there are two simple saddle points in the bulk: $x_\pm=\frac{-iu\pm\sqrt{1-u^{2}}}{2}$  where $u\in (-1,1)$.  At the rightmost soft edge (which corresponds to $u= 1$), these two points coincide and give a single saddle point $x_0=-\frac{i}{2}$ of multiplicity 2.

In the following lines, we proceed to compute the leading asymptotic terms of the weighted expectation $\varphi_{N}(s_1,\ldots,s_n)$.  We start with the bulk of the spectrum and follow  the soft edge.

\begin{proof}[Proof of Theorem \ref{theobulk}: $\HBE$ case] Set $u=\sin\theta,\ \theta\in(-\frac{\pi}{2},\frac{\pi}{2})$, so that
$x_+=\tfrac{1}{2}e^{-i\theta}$, $x_-=\tfrac{1}{2}e^{i(\theta+\pi)}$ and
$$p(x_+)=-\frac{1}{2}\cos2\theta+(1+\ln2)+i(\theta+\frac{1}{2}\sin2\theta),\quad  p(x_-)=-\frac{1}{2}\cos2\theta+(1+\ln2)-i(\theta+\frac{1}{2}\sin2\theta+\pi).$$
It follows from $$p_\pm:=p''(x_\pm)=8e^{\pm i\theta}\cos\theta$$
that the angles of steepest descent are $-\theta/2$ at $x_+$ and  $\theta/2$ at $x_-$. Note that if $x=-\frac{i}{2}u+v,\ v\in \mathbb{R}$, then the function
$$\textrm{Re}\{p(x)\}=2v^{2}-\frac{1}{2}\ln(v^{2}+\frac{1}{4}u^{2})+\frac{3}{2}u^{2}$$
attains its minimum value at the point $v=\pm \frac{1}{2}\sqrt{1-u^{2}}$. Therefore, as a possible path, we consider the straight line passing through $-\frac{i}{2}u$ and parallel to the real axis (together with irrelevant deformations at $\pm \infty$).

   The   segments of integration $(-\infty e^{-i\frac{\theta}{2}},  \infty e^{-i\frac{\theta}{2}})$ and $(-\infty e^{i\frac{\theta}{2}},  \infty e^{i\frac{\theta}{2}})$  can be deformed into the real axis  near the saddle points.  Since  $x_+ -x_->0$,   the assumption (vi) is fulfilled. Take $\rho=\frac{2}{\pi}\sqrt{1-u^{2}}$.  According to   Corollary \ref{twosaddle}, for $n=2m$ ($l=0$),
   \begin{align} I_{N,2m}&\sim (8N)^{-\beta'\! m(m-1)/2-m}\exp\{-nNu^{2} -nN(1+2\ln 2-i\pi)/2\}\, (\frac{\pi \rho}{2})^{\beta'\! m(m+1)/2-m} \nonumber\\
&\times \tbinom{2m}{m} (\Gamma_{\beta'\!,m})^{2} \,{\!\!\!{\phantom{j}}_{0}\mathcal{F}_0^{(2/\beta'\!)}}( i\pi s;1^{m},(-1)^{m})
,\end{align}
while for $n=2m-1$,
\begin{align} &I_{N,2m-1}\sim (8N)^{-\beta'\! (m-1)^{2}/2-n/2}\exp\{-nNu^{2} -nN(1+2\ln 2-i\pi)/2\}\, (\frac{\pi \rho}{2})^{\beta'\! (m^{2}-1)/2-n/2} \nonumber\\
&\times \tbinom{2m-1}{m} \Gamma_{\beta'\!,m-1} \Gamma_{\beta'\!,m}\, (-2i)^{(2m-1)l}
\Big(
e^{i\theta_{N}-il(\theta+\frac{\pi}{2})} E_{m-1}^{(2/\beta'\!)}( i\pi s)+e^{-i\theta_{N}+il(\theta+\frac{\pi}{2})} E_{m-1}^{(2/\beta'\!)}( -i\pi s)
\Big),\end{align}
where  $$\theta_{N}=N(2\theta+\sin2\theta+\pi)/2+\theta(1+(m-1)\beta')/2, \qquad \theta=\arcsin u.$$ Hence,
for $n=2m$,
   \begin{align} \varphi_{N}\big(\sqrt{2N}(u+\frac{s}{\rho N})\big)&\sim \Psi_{\!^{N,2m}} \gamma_{m}(\beta'\!)\ {\!\!\!{\phantom{j}}_{0}\mathcal{F}_0^{(2/\beta'\!)}}( i\pi s;1^{(m)},(-1)^{(m)}) \label{evenbulkHE}
,\end{align}
where
\be\Psi_{\!^{N,2m}}=(\pi \rho)^{\beta'\! m(m+1)/2-m} N^{\beta'\! m^{2}/2}\exp\{ -m N(1+\ln 2-\ln N)\} \nonumber\ee
and \be \gamma_{m}(\beta'\!):=\binom{2m}{m}\prod_{j=1}^m\frac{\Gamma(1+\beta'\!j/2)}{\Gamma(1+\beta'\!(m+j)/2)}.\ee
For $n=2m-1$,
\begin{align} \varphi_{N-l}\big(\sqrt{2N}(u+\frac{s}{\rho N})\big)&\sim \Psi^{\!_{(l)}}_{\!^{N,n}}  \frac{1}{2i\sqrt{\cos\theta}}\Big(
e^{i\theta_{N}-il(\theta+\frac{\pi}{2})} E_{m-1}^{(2/\beta'\!)}( i\pi s)+e^{-i\theta_{N}+il(\theta+\frac{\pi}{2})} E_{m-1}^{(2/\beta'\!)}( -i\pi s)
\Big) \label{oddbulkHE}
,\end{align}
where
\begin{multline*}\Psi^{\!_{(l)}}_{\!^{N,2m-1}}=\tbinom{2m-1}{m} \Gamma_{\beta'\!,m-1} \Gamma_{\beta'\!,m}(\Gamma_{\beta'\!,2m-1})^{-1}(\pi \rho)^{\beta'\!(m^{2}-1)/2-(2m-1)/2} N^{\beta'\! m(m-1)/2} \\ \times   \exp\{ -(2m-1) N(1+\ln 2-\ln N)/2\}\, (\sqrt{N/2})^{-(2m-1)l} (2i\sqrt{\pi \rho/2}). \end{multline*}\end{proof}

\begin{proof}[Proof of Theorem \ref{theosoft}: H$\beta$E case] In Eqs.\eqref{chHE1}--\eqref{chHE3}, let $u=1$, $l=0$ and $\rho=2N^{-1/3}$.
Then,  \be p(x)=2x^{2}+4i x-\ln x, \qquad q(t)= {\!\!\!{\phantom{j}}_{0}\mathcal{F}_0^{(2/\beta'\!)}}\!(-2iN^{1/3}s ;t+i/2)\nonumber.\ee
At the double saddle point $x_0=-i/2$, we have
$$p(x_0)=3/2+\ln2+i \pi/2,\qquad p'''(x_0)=16i.$$ The angle of steepest descent   are thus $-5\pi/6$ and $-\pi/6$.  However, we see that $ \textrm{Re}\{p(x+x_0)-p(x_0)\}>0$ for all $x\in\mathbb{R}$ except at the origin.  We thus choose $\mathscr{P}$ to follow the real line from $-\infty$ to $\infty$.  Corollary \ref{onesaddle} then implies
 \begin{align} I_{N,n}&\sim (8N)^{-\beta'\! n(n-1)/6-n/3}\exp\{ -nN(3+2\ln 2+i\pi)/2\}\,(2\pi)^{n} \textrm{Ai}^{(2/\beta'\!)}(s).\end{align}
Thus,
  \begin{align} \varphi_{N}\big(\sqrt{2N}(1+\frac{s}{\rho N})\big)&\sim \Phi_{\!^{N,n}}\, (2\pi)^{n}(\Gamma_{\beta'\!,n})^{-1} \textrm{Ai}^{(2/\beta'\!)}(s) ,\label{softHE}\end{align}
where
\be\Phi_{\!^{N,n}}= N^{\beta'\! n(n-1)/12+n/6}\exp\{ -n N(1+\ln 2-\ln N-2i\pi)/2\}. \nonumber\ee\end{proof}

We now turn our attention to the eigenvalue correlation functions (marginal densities). For $\beta$ even,  the scaling limits of the correlation functions for the  H$\beta$E immediately follow  from  \eqref{evenbulkHE} and \eqref{softHE}. Let $n=2m=k\beta$, we have the following relation
\be\label{eqcorr} R_{k,N}(x_1,\ldots,x_k)=\frac{(k+N)!}{N!}\frac{G_{\beta,N}}{G_{\beta,k+N}}   \prod_{1\leq j<l\leq
k}(x_{j}-x_{l})^{\beta}\,\big[\varphi_N(s_1,\ldots,s_n)\big]_{\{s\}\mapsto \{x\}}.
\ee
One easily shows that as $N\to\infty$,
\be\label{eqasympconst}\frac{G_{\beta\!,N}}{G_{\beta\!,k+N}}\sim (2\pi\!)^{-k/2}2^{\beta k(k+1)/4}\beta^{-\beta k/2} (\Gamma(1+\beta/2))^{k}(2e)^{\beta kN/2}N^{-\beta kN/2-\beta k(k+1)/4-k/2}.\ee

\begin{proof}[Proof of Corollaries \ref{theobulkCF} and \ref{theosoftCF}: H$\beta$E case]
First consider the bulk scaling  and let $\rho=\frac{2}{\pi}\sqrt{1-u^{2}}$.  The use of \eqref{eqcorr}, \eqref{eqasympconst}, and \eqref{evenbulkHE}, then  leads to
\begin{align} \big(\frac{\sqrt{2N}}{\rho N}\big)^{k} R_{k,N}\big(\sqrt{2N}(u+\frac{x}{\rho N})\big)\sim  (\beta/2)^{-\beta k/2}(\Gamma(1+\beta/2))^{k} \gamma_{m}(\beta'\!)\, |\Delta(2\pi x)|^{\beta} {\!\!\!{\phantom{j}}_{0}\mathcal{F}_0^{(2/\beta'\!)}}( i\pi s;1^{m}\!,(-1)^{m}).\nonumber\end{align}
According to  Gauss's multiplication formulas for the gamma function \cite{aar},
$$\prod_{j=1}^{l} \Gamma\big(a+  \frac{j-1}{l}\big)=l^{-la+\frac{1}{2}}(2\pi)^{\frac{l-1}{2}} \Gamma(la), \qquad l\in \mathbb{N}, $$
we have for $l=\beta/2$,
$$\gamma_{m}(\beta'\!):=\binom{2m}{m}\prod_{j=1}^m\frac{\Gamma(1+\beta'\!j/2)}{\Gamma(1+\beta'\!(m+j)/2)}=(\beta/2)^{\beta k^{2}/2} \prod_{j=0}^{k-1}\frac{\Gamma(1+\beta j/2)}{\Gamma(1+\beta (k+j)/2)}\, .$$
This further implies
\begin{align} \big(\frac{\sqrt{2N}}{\rho N}\big)^{k} R_{k,N}\big(\sqrt{2N}(u+\frac{x}{\rho N})\big)\sim b_{k}(\beta) \, |\Delta(2\pi x)|^{\beta} {\!\!\!{\phantom{j}}_{0}\mathcal{F}_0^{( \beta/2)}}( i\pi s;1^{m}\!,(-1)^{m})_{\{s\}\mapsto \{x\}}\, ,\end{align}where \be b_{k}(\beta):=  (\beta/2)^{\beta k(k-1)/2}(\Gamma(1+\beta/2))^{k} \prod_{j=0}^{k-1}\frac{\Gamma(1+\beta j/2)}{\Gamma(1+\beta (k+j)/2)} \, , \ee
which establishes the first part of  Corollary \ref{theobulkCF}.

For the soft edge, one first sets   $\rho=2N^{-1\!/3}$.  Then, recalling \eqref{eqcorr}, \eqref{eqasympconst}, and \eqref{softHE}, one finds
\begin{align} \big(\frac{\sqrt{2N}}{\rho N}\big)^{k} R_{k,N}\big(\sqrt{2N}(1+\frac{x}{\rho N})\big)\sim  a_{k}(\beta)\, |\Delta(x)|^{\beta} \textrm{Ai}^{\!(\beta/2)}(s)_{\{s\}\mapsto \{x\}},\end{align}
where \be a_{k}(\beta):=  (\beta/2)^{(\beta k+1)k}(\Gamma(1+\beta/2))^{k} \prod_{j=1}^{2k} \frac{(\Gamma(1+2/\beta))^{\beta /2}}{\Gamma(1+\beta j/2) },   \ee
and the proof of the first half of Corollary \ref{theosoftCF} is completed.\end{proof}

Note that the coefficient  $b_{k}(\beta)$ given above
is exactly same as that in the circular $\beta$-ensemble (Proposition 13.2.3, \cite{forrester}) as well as in the L$\beta$E and  the J$\beta$E below.  This   strongly suggests the universality of $b_{k}(\beta)$.

\subsection{Laguerre $\beta$-ensemble}
  Based on the work of Kaneko \cite{kaneko}, it is easy to verify that  if
\be K_N(s_1,\ldots,s_n)=\Big\langle \prod_{i=1}^N\prod_{j=1}^n(x_i-s_j)\Big\rangle_{\LBE},\nonumber \ee
then
\be K_N(s_1,\ldots,s_n)=\frac{
W_{\lambda_1+n, \beta,N}
}{W_{\lambda_1, \beta,N}}
 \!\!{\phantom{j}}_1 F_{1}^{(\beta/2)}\left(-N;(2/\beta)(\lambda_1+n); s_1,\ldots,s_n\right), \ee
for instance, see Proposition 13.2.5,\cite{forrester}.
By making the change $s_i\mapsto s_i/N$ and taking the limit $N\to\infty$, one readily proves the first formula of Theorem \ref{theohard}.  Note that this result could be obtained by taking the integral duality formula in \cite{des} and following the asymptotic method developed in the previous section for the case where $p(x)$ admits a simple saddle point.

By using \eqref{forrester11}, we get the following integral formula
\be   K_N(s_1,\ldots,s_n)=  A_{N}\, i^{-n}
\int_{\mathbb{T}^{n}} \prod_{j=1}^n t_j^{-N-1}(1-t_j)^{N-1+\beta' (\lambda_{1}+1)/2} |\Delta(t)|^{\beta'} {\!\!\!{\phantom{j}}_0\mathcal{F}_{0}^{(2/\beta')}(s;t)}\,d^{n}t
\ee
where
\be  A_{N}= \frac{(2\pi)^{-n} e^{i\pi n N}}{M_n(N,-1+\beta' (\lambda_{1}+1)/2,\beta'/2)} \frac{
W_{\lambda_1+n, \beta,N}
}{W_{\lambda_1, \beta,N}}.\ee

Now let
 $$\varphi_{N}(s_1,\ldots,s_n):=e^{-\frac{1}{2}p_{1}(s)}\prod_{1\leq j\leq n}\!s_{j}^{\lambda_{1}/\beta} \,K_N(s_1,\ldots,s_n).$$
The
application of Proposition \ref{proppractical} then yields  \be \varphi_{N-l}\big(4N(u+\frac{s}{\rho N})\big)=A_{N-l}\, (4N)^{n\lambda_1 /\beta}e^{-2nNu}\!\prod_{1\leq j\leq n}\!\big(u+\frac{s_{j}}{\rho N}\big)^{\lambda_{1}/\beta} \,I_{N,n},\label{chLE1}\ee
where \begin{align} I_{N,n}&=\int_{\mathbb{T}^n} \exp\{-N\sum\limits_{j=1}^n p(t_j)\}|\Delta(t)|^{\beta'\!}  q(t) d^{n}t\nonumber\\
&=
c_{\beta'\!} \int_{0}^{\infty}r^{-\beta'\! -1}\bigg(\int_{\mathbb{T}^n}
  \exp\{-N\sum\limits_{j=1}^{n}p(t_{j})\}\, q(t)\,H_{\beta'\!}\big(r \Delta(t)\prod_{j=1}^{n}(it_j)^{-(n-1)/2}\big)\, d^{n} t\bigg) d r,\label{chLE2}\end{align}
  and
\be p(x)=\ln x-\ln(1-x)-4u x, \qquad q(t)=i^{-n}\prod_{j=1}^n t_j^{l-1}(1-t_j)^{-l-1+\beta' (\lambda_{1}+1)/2}\,{\!\!\!{\phantom{j}}_{0}\mathcal{F}_0^{(2/\beta'\!)}}(4s/\rho;t-1/2)\label{chLE3}.\ee

Since $p'(x)=\tfrac{1}{x(1-x)}-4u$, there are two simple saddle points in the bulk, namely  $x_\pm=\tfrac{1}{2}\pm  \tfrac{i}{2}\sqrt{\tfrac{1-u}{u}}$  with $u\in (0,1)$.  By letting $u\to1$, which corresponds to the soft edge,  we find that the two saddle points  become  one double saddle point $x_0=\tfrac{1}{2}$.

\begin{proof}[Proof of Theorem \ref{theobulk}: L$\beta$E case]
Here we focus on the bulk of the spectrum, so we set $u=\cos^{2}\theta,\ \theta\in(0,\tfrac{\pi}{2})$. Hence
$x_\pm=\tfrac{1}{2\cos\theta}e^{\pm i\theta},$ and
$$p(x_\pm)=-2\cos^{2}\theta\pm i(2\theta-\sin2\theta),\qquad p_\pm:=p''(x_\pm)=\pm 16iu^{2}\sqrt{\tfrac{1-u}{u}}.$$
We see that the angles of steepest descent are $\tfrac{3}{4}\pi$ at $x_+$ and  $\tfrac{1}{4}\pi$ at $x_-$, so we choose the following path of integration:
$$\mathscr{L}_{1}^{1\!/\!(2\!\cos\!\theta)}\bigcup\mathscr{M}_{1\!/\!(2\!\cos\!\theta)}^{-1\!/\!(2\!\cos\!\theta)}
\bigcup\mathscr{M}_{-1\!/\!(2\!\cos\!\theta)}^{1\!/\!(2\!\cos\!\theta)}\bigcup \mathscr{L}_{1\!/\!(2\!\cos\!\theta)}^{1}. $$
Actually, when  $x\in \mathscr{L}_{1}^{1\!/\!(2\!\cos\!\theta)}$ or $\mathscr{L}_{1\!/\!(2\!\cos\!\theta)}^{1}$, we have
$$\textrm{Re}\{p(x)-p(x_+)\}=\ln\frac{x}{ |1-x|}-4ux+2u>0. $$
On the semicircle $\mathscr{M}_{1\!/\!(2\!\cos\!\theta)}^{-1\!/\!(2\!\cos\!\theta)}$, we can write $x=\tfrac{1}{2\cos\theta}e^{i\phi}$ with  $\phi\in [0,\pi]$. Since
$$g(\cos\phi):=\textrm{Re}\{p(x)-p(x_+)\}=-\tfrac{1}{2}\ln(1+4\cos^{2}\theta-4\cos\theta\cos\phi)+2(\cos\theta-\cos\phi)\cos\theta,$$
it follows from $$g'(\cos\phi)=\frac{8(\cos\phi-\cos\theta)\cos^{2}\theta}{ 1+4\cos^{2}\theta-4\cos\theta\cos\phi}$$
that $g(\cos\phi)$
attains its minimum value at $\cos\phi=\cos\theta$. Similar results hold for $\textrm{Re}\{p(x)-p(x_-)\}$.

Case 1: $n=2m$.

   The line segments of integration $(-\infty e^{i\frac{3}{4}\pi},  \infty e^{i\frac{3}{4}\pi})$ and $(-\infty e^{i\frac{1}{4}\pi},  \infty e^{i\frac{1}{4}\pi})$   become the real line by  making the following change of variables: $w_j\mapsto w_j e^{i\frac{3}{4}\pi}$ and $w_j\mapsto w_j e^{i\frac{1}{4}\pi}$, respectively near $x_+$ and $x_-$. We take  $\rho=\frac{2}{\pi}\sqrt{\tfrac{1-u}{u}}$ and $l=0$.  Then,  according to   Theorem \ref{teo2saddle},  we have
   \begin{align} I_{N,2m}&\sim N^{-\beta'\! m(m-1)/2-m}e^{4mNu}(2\sqrt{u})^{-\beta'\!m\lambda_1}\, (\frac{\pi \rho}{2})^{\beta'\! m(m+1)/2-m} \nonumber\\
&\times \tbinom{2m}{m} (\Gamma_{\beta'\!,m})^{2} \,{\!\!\!{\phantom{j}}_{0}\mathcal{F}_0^{(2/\beta'\!)}}( i\pi s;1^{m},(-1)^{m}).\end{align}
From $$ \frac{
W_{\lambda_1+n, \beta,N}
}{W_{\lambda_1, \beta,N}}=(2/\beta)^{nN}\prod_{j=0}^{N-1}\frac{\Gamma(1+\lambda_{1}+n+\beta j/2)}{\Gamma(1+\lambda_{1}+\beta j/2)}=\prod_{j=1}^{n}\frac{\Gamma(N+\beta'\!(\lambda_{1}+ j)/2)}{\Gamma(\beta'\!(\lambda_{1}+ j)/2)},$$
we also get
\begin{align}A_N&=(2\pi)^{-n} e^{i\pi n N}\prod_{j=1}^{n}\frac{\Gamma(1+\beta'\!/2)\Gamma(1+N+\beta'\!(n-j)/2)}{\Gamma(1+\beta'\!j/2)}\nonumber\\
&\sim   (\Gamma_{\beta'\!,n})^{-1}N^{\beta'\! n(n-1)/4+n/2}\exp\{-nN(1-\ln N-i\pi)\}. \end{align}
Hence,
\begin{align} \varphi_{N}\big(4N(u+\frac{s}{\rho N})\big)&\sim \Psi_{\!^{N,2m}} \gamma_{m}(\beta'\!)\ {\!\!\!{\phantom{j}}_{0}\mathcal{F}_0^{(2/\beta'\!)}}( i\pi s;1^{m},(-1)^{m}) \label{evenbulkLE}
,\end{align}
where
\be\Psi_{\!^{N,2m}}=(\pi \rho/2)^{\beta'\! m(m+1)/2-m} N^{\beta'\! m(m+\lambda_1)/2}\exp\{ -2m N(1-\ln N)\}. \nonumber\ee

Case 2: $n=2m-1$. According to Theorem \ref{teo2saddle}, two types of integration domains  lead to a significant contribution as $N\to \infty$.  They correspond to
the intervals  of integration $$(-\infty e^{i\frac{3}{4}\pi},  \infty e^{i\frac{3}{4}\pi})^{m-1} \times (-\infty e^{i\frac{1}{4}\pi},  \infty e^{i\frac{1}{4}\pi})^{m}, \qquad (-\infty e^{i\frac{3}{4}\pi},  \infty e^{i\frac{3}{4}\pi})^{m} \times (-\infty e^{i\frac{1}{4}\pi},  \infty e^{i\frac{1}{4}\pi})^{m-1}$$
    which can be deformed to (whenever $m>1$)
  $$(-\infty e^{i\frac{3}{4}\pi},  \infty e^{i\frac{3}{4}\pi})^{m-1} \times (-\infty e^{\frac{i}{4}\pi+\frac{i}{m}(\frac{\pi}{2}-2\theta)},  \infty e^{\frac{i}{4}\pi+\frac{i}{m}(\frac{\pi}{2}-2\theta)})^{m}$$
  and $$(-\infty e^{i\frac{3}{4}\pi+\frac{i}{m}(2\theta-\frac{\pi}{2})},  \infty e^{i\frac{3}{4}\pi+\frac{i}{m}(2\theta-\frac{\pi}{2})})^{m} \times (-\infty e^{i\frac{1}{4}\pi},  \infty e^{i\frac{1}{4}\pi})^{m-1}.$$
Then, by changing the variables as in $n=2m$ case, we get
\begin{align} &I_{N,2m-1}\sim N^{-\beta'\! (m-1)^{2}/2-n/2}\exp\{2nNu \}\, (\frac{\pi \rho}{2})^{\beta'\! (m^{2}-1)/2-n/2} (2\sqrt{u})^{-\beta'\!(1+n\lambda_1)/2} \nonumber\\
&\times \tbinom{2m-1}{m} \Gamma_{\beta'\!,m-1} \Gamma_{\beta'\!,m} \frac{1}{i}
\Big(
e^{i(\theta_{N}-2l \theta)} E_{m-1}^{(2/\beta'\!)}(- i\pi s)-e^{-i(\theta_{N}-2l \theta)} E_{m-1}^{(2/\beta'\!)}( i\pi s)
\Big),\end{align}
where $$\theta_{N}=N(2\theta-\sin2\theta) +\beta'\! (\lambda_1+1)\theta/2+\beta' (m-1)(\theta-\pi/4)+\pi/4, \qquad \theta=\arccos \sqrt{u}.$$
One finally obtains  \begin{align} \varphi_{N-l}\big(4N(u+\frac{s}{\rho N})\big)&\sim \Psi^{\!_{(l)}}_{\!^{N,n}} \frac{1}{2\sqrt{2}i \sqrt[4]{u(1-u)} }\Big(
e^{i(\theta_{N}-2l \theta)} E_{m-1}^{(2/\beta'\!)}(- i\pi s)-e^{-i(\theta_{N}-2l \theta)} E_{m-1}^{(2/\beta'\!)}( i\pi s)
\Big) \label{oddbulkLE}
,\end{align}
where
\begin{align*}\Psi^{\!_{(l)}}_{\!^{N,2m-1}}&=\tbinom{2m-1}{m} \Gamma_{\beta'\!,m-1} \Gamma_{\beta'\!,m}(\Gamma_{\beta'\!,2m-1})^{-1}(\pi \rho/2)^{\beta'\!(m^{2}-1)/2- m+1} \sqrt{2}(2\sqrt{u})^{-\beta'\!/2+1}  \\ &\times  N^{\beta'\! m(m-1)/2+\beta'\! (2m-1)\lambda_1/4} \exp\{ -(2m-1) N(1-\ln N-i\pi)\}\, (-N)^{-(2m-1)l}. \end{align*}
\end{proof}

\begin{proof}[Proof of Theorem \ref{theosoft}: L$\beta$E case]
We now consider    the soft edge of the spectrum ($u=1$). In Eqs.\eqref{chLE1}--\eqref{chLE3}, let $u=1$, $l=0$ and $\rho=2(2N)^{-1/3}$.  This yields
  \be p(x)=\ln x-\ln(1-x)-4x, \qquad q(t)=i^{-n}\prod_{j=1}^n t_j^{-1}(1-t_j)^{-1+\beta' (\lambda_{1}+1)/2}\,{\!\!\!{\phantom{j}}_{0}\mathcal{F}_0^{(2/\beta'\!)}}(4s/\rho;t-1/2)\nonumber.\ee
At the double saddle point $x_0=1/2$, we have
$$p(x_0)=-2,\qquad  p'''(x_0)=32.$$ Thus, the angles of steepest descent   are $2\pi/3$ and $4\pi/3$.   Since on the circle $x=\tfrac{1}{2}e^{i\phi}$, we have
$$ \textrm{Re}\{p(x+x_0)-p(x_0)\}=-\ln\sqrt{5-4\cos\phi}-2\cos\phi+2>0$$
whenever $\phi\in (0,2\pi)$.
We choose $\mathscr{P}$ to be the following path: it starts at 1, arrives at $1/2+i0^{+}$ by following a straight line, along the centered  circle of radius $1/2$ counterclockwisely,   then follows a straight line from $ 1/2-i0^{+}$ to  $1$.  Applying Theorem \ref{mainintegral1} and Remark \ref{remarkmu},  noting that the line segments of integration near the saddle point can be chosen  to be  $(-\infty i,  \infty i)$,    after the change of variables: $w_j\mapsto 16^{-1/3}i w_j$ one obtains
 \begin{align} I_{N,n}&\sim  2^{-\beta'\! n(n-1)/6-\beta'\! n(\lambda_1+1)/2+2n/3} N^{-\beta'\! n(n-1)/6-n/3} \exp\{ 2nN\}\,(2\pi)^{n} \textrm{Ai}^{(2/\beta'\!)}(s).\end{align}
Thus,
  \begin{align} \varphi_{N}\big(4N(1+\frac{s}{\rho N})\big)&\sim \Phi_{\!^{N,n}}\, (2\pi)^{n}(\Gamma_{\beta'\!,n})^{-1} \textrm{Ai}^{(2/\beta'\!)}(s) ,\label{softLE}\end{align}
where
\be\Phi_{\!^{N,n}}=  2^{-\beta'\! n(n-1)/6-\beta'\! n/2+2n/3} N^{\beta'\! n(n-1)/12+\beta'\! n \lambda_1/4+ n/6}\exp\{ -n N(1-\ln N-i\pi)\}. \nonumber\ee
\end{proof}

\begin{proof}[Proof of Corollaries \ref{theobulkCF} and \ref{theosoftCF}: L$\beta$E case]
As mentioned previously, when $\beta$ is even,  the scaling limit of the correlation functions for the  L$\beta$E immediately follow  from  \eqref{evenbulkLE} and \eqref{softLE}. For the bulk case, let $n=2m=k\beta$ and $\rho=\frac{2}{\pi}\sqrt{\tfrac{1-u}{u}}$.  By making use of \be R_{k,N}(x_1,\ldots,x_k)=\frac{(k+N)!}{N!}\frac{W_{\lambda_1,\beta,N}}{W_{\lambda_1,\beta,k+N}}   \prod_{1\leq j<l\leq
k}(x_{j}-x_{l})^{\beta}\,\big[\varphi_N(s_1,\ldots,s_n)\big]_{\{s\}\mapsto \{x\}}
\ee
and  $$\frac{W_{\lambda_1,\beta,N}}{W_{\lambda_1,\beta,k+N}}\sim (2\pi\!)^{-k} (\beta/2)^{-\beta k/2} (\Gamma(1+\beta/2))^{k} e^{\beta kN}N^{-\beta kN -\beta k^{2}/2-(\lambda_1+1)k},$$
one can show that
\begin{align} \big(\frac{4}{\rho}\big)^{k} R_{k,N}\big(4N(u+\frac{x}{\rho N})\big)\sim b_{k}(\beta) \, |\Delta(2\pi x)|^{\beta} {\!\!\!{\phantom{j}}_{0}\mathcal{F}_0^{( \beta/2)}}( i\pi s;1^{m}\!,(-1)^{m})_{\{s\}\mapsto \{x\}}
.\end{align}
For   the soft edge,  one   lets $\rho=2(2N)^{-1\!/3}$ and gets
\begin{align} \big(\frac{4}{\rho}\big)^{k} R_{k,N}\big(4N(1+\frac{x}{\rho N})\big)\sim  a_{k}(\beta)\, |\Delta(x)|^{\beta} \textrm{Ai}^{\!(\beta/2)}(s)_{\{s\}\mapsto \{x\}}.\end{align}
Notice that the coefficients $ a_{k}(\beta)$ and $b_{k}(\beta)$ are the same as those in the H$\beta$E.
\end{proof}

\subsection{Jacobi $\beta$-ensemble } The J$\beta$E case is very similar to the L$\beta$E. First, Kaneko \cite{kaneko} proved that
\begin{multline} K_N(s_1,\ldots,s_n):=\Big\langle \prod_{i=1}^N\prod_{j=1}^n(x_i-s_j)\Big\rangle_{\JBE}\\
=\frac{
S_{N}(\lambda_{1}+n,\lambda_{2},\beta/2)
}{S_{N}(\lambda_{1},\lambda_{2},\beta/2)}\,
\!\!\!{\phantom{j}}_2 F_{1}^{(\beta/2)}\left(-N,(2/\beta)(\lambda_1+\lambda_2+n+1)+N-1;(2/\beta)(\lambda_1+n); s\right). \end{multline}
 By making the change $s_i\mapsto s_i/N^{2}$ and taking the limit $N\to\infty$, one readily proves the second formula of Theorem \ref{theohard}.

By using \eqref{forrester21} we get the following integral formula
\be   K_N(s)=  B_{N}\, i^{-n}
\int_{\mathbb{T}^{n}} \prod_{j=1}^n t_j^{-\beta'\!(\lambda_2+1)/2-N}(1-t_j)^{N-2+\beta'\! (\lambda_{1}+\lambda_{2}+2)/2} |\Delta(t)|^{\beta'} {\!\!\!{\phantom{j}}_1\mathcal{F}_{0}^{(2/\beta')}(-N;s;1-t)}\,d^{n}t
\ee
where
\be  B_{N}= \frac{(2\pi)^{-n} e^{i\pi(\beta'\! n(\lambda_{2}+1)/2+n(N-1))}}{M_n(\beta' (\lambda_{2}+1)/2+N-1,-1+\beta' (\lambda_{1}+1)/2,\beta'/2)} \frac{
S_{N}(\lambda_{1}+n,\lambda_{2},\beta/2)
}{S_{N}(\lambda_{1},\lambda_{2},\beta/2)}.\ee

For the weighted quantity
 $$\varphi_{N}(s_1,\ldots,s_n):= \prod_{1\leq j\leq n}\!s_{j}^{\lambda_{1}/\beta}(1-s_{j})^{\lambda_{2}/\beta} \,K_N(s_1,\ldots,s_n),$$
application of Proposition \ref{proppractical} gives  \be \varphi_{N-l}\big( u+\frac{s}{\rho N} \big)=B_{N-l}\, \prod_{1\leq j\leq n}\!\big(u+\frac{s_{j}}{\rho N}\big)^{\lambda_{1}/\beta}\big(1-u-\frac{s_{j}}{\rho N}\big)^{\lambda_{2}/\beta} \,I_{N,n},\label{chJE1}\ee
where \begin{align} I_{N,n} =
c_{\beta'\!} \int_{0}^{\infty}r^{-\beta'\! -1}\bigg(\int_{\mathbb{T}^n}
  \exp\{-N\sum\limits_{j=1}^{n}p(t_{j})\}\, q(t)\,H_{\beta'\!}\big(r \Delta(t)\prod_{j=1}^{n}(it_j)^{-(n-1)/2}\big)\, d^{n} t\bigg) d r,\label{chJE2}\end{align}
  and $p(x)=\ln x-\ln(1-x)-\ln(1-u+ux)$,
\be q(t)=i^{-n}\prod_{j=1}^n t_j^{-\beta'\!(\lambda_2+1)/2+l}(1-t_j)^{\beta'\! (\lambda_{1}+\lambda_{2}+2)/2-l-2}\,{\!\!\!{\phantom{j}}_{1}\mathcal{F}_0^{(2/\beta'\!)}}(-N;\frac{s}{\rho N};\frac{1-t}{1-u+ut})\label{chJE3}.\ee

Since
 $p'(x)=\tfrac{1}{x(1-x)}- \frac{u}{1-u+ux}$,    there are generically two simple saddle points:
 $$x_+=\sqrt{(1-u)/u}\,e^{i \pi/2}\quad \text{and}\quad   x_-=\sqrt{(1-u)/u}\,e^{3i\pi/2},$$
in the bulk of the spectrum with $u\in (0,1)$.

\begin{proof}[Proof of Theorem \ref{theobulk}: J$\beta$E case]
 Set $u=\cos^{2}\theta,\ \theta\in(0, \pi/2)$, then
$$p(x_\pm)= \pm 2i\theta,\qquad p_\pm:=p''(x_\pm)=2u^{2}/\sqrt{u(1-u)}\  e^{\pm i(2\theta- \pi/2)}.$$
This allows us to take the angles of steepest descent $\tfrac{5}{4}\pi-\theta$ at $x_+$ and  $\theta-\tfrac{1}{4}\pi$ at $x_-$. We choose the path of integration:
$$\mathscr{L}_{1}^{\tan\!\theta}\bigcup\mathscr{M}_{\tan\!\theta}^{-\!\tan\!\theta}
\bigcup\mathscr{M}_{-\!\tan\!\theta}^{\tan\!\theta}\bigcup \mathscr{L}_{\tan\!\theta}^{1}. $$
Actually, when  $x\in \mathscr{L}_{1}^{\tan\!\theta}$ or $\mathscr{L}_{\tan\!\theta}^{1}$, we see that
$$\textrm{Re}\{p(x)-p(x_\pm)\}=\ln\frac{x}{ |1-x|}-\ln(1-u+ux)>0,$$
while on the  circle $\{x:|x|=\tan\!\theta\}$, %%$\mathscr{M}_{\tan\!\theta}^{-\!\tan\!\theta} \bigcup \mathscr{M}_{-\tan\!\theta}^{\tan\!\theta}$
 setting $x=\tan\!\theta\, e^{i\phi}$ with $\phi\in [0,2\pi)$,
$$ \textrm{Re}\{p(x)-p(x_\pm)\}=-\tfrac{1}{2}\ln(1-4u(1-u)\cos^{2}\phi) $$
attains its minimum at $ \phi=\pm \pi/2$.

Let $\rho=\tfrac{1}{\pi}\tfrac{1}{\sqrt{u(1-u)}}$, notice the polynomial
$$
{\!\!\!{\phantom{j}}_{1}\mathcal{F}_0^{(2/\beta'\!)}}\big(-N;\frac{s}{\rho N};\frac{1-t}{1-u+ut}\big)={_{0}\mathcal{F}_0^{(2/\beta'\!)}} \big(\frac{s}{\rho };\frac{ t-1}{1-u+ut}\big) +O(N^{-1}),$$
and $$B_N \sim   (\Gamma_{\beta'\!,n})^{-1} 2^{-\beta'\! n(\lambda_{1}+\lambda_{2}+2)/2-\beta'\! n(n-1)/4-n(2N-3/2)} N^{\beta'\! n(n-1)/4+n/2}\exp\{i\pi\big(\beta'\! n(\lambda_{2}+1)/2+n(N-1)\big) \},$$
we are ready to compute the bulk scaling.

For  $n=2m$,
\begin{align} \varphi_{N}\big( u+\frac{s}{\rho N} \big)&\sim \Psi_{\!^{N,2m}} \gamma_{m}(\beta'\!)\ {\!\!\!{\phantom{j}}_{0}\mathcal{F}_0^{(2/\beta'\!)}}( i\pi s;1^{m},(-1)^{m}) \label{evenbulkJE}
,\end{align}
where
\be\Psi_{\!^{N,2m}}=(\pi \rho)^{\beta'\! m(m+1)/2-m} N^{\beta'\! m^{2}/2}2^{-\beta'\! m^{2}/2-\beta'\! m(\lambda_{1}+\lambda_{2}+1)+2m(1-2N)}, \nonumber\ee
while for $n=2m-1$ \begin{align} \varphi_{N-l}\big( u+\frac{s}{\rho N} \big)&\sim \Psi^{\!_{(l)}}_{\!^{N,n}} \frac{1}{2  \sqrt[4]{u } }\Big(
e^{i \theta_{N} +i(\frac{\pi}{2}-\theta)l} E_{m-1}^{(2/\beta'\!)}(- i\pi s)-e^{-i \theta_{N}-i(\frac{\pi}{2}-\theta)l} E_{m-1}^{(2/\beta'\!)}( i\pi s)
\Big) \label{oddbulkJE}
,\end{align}
where
\begin{align*}\Psi^{\!_{(l)}}_{\!^{N,2m-1}}&=\tbinom{2m-1}{m} \Gamma_{\beta'\!,m-1} \Gamma_{\beta'\!,m}(\Gamma_{\beta'\!,2m-1})^{-1}(\pi \rho)^{\beta'\!(m-1)^{2}/2 +(\beta'\!-2m+1)/2} \sqrt{2}(2\sqrt{u})^{-\beta'\!/2+1}  \\ &\times  N^{\beta'\! m(m-1)/2} i^{-1}(-1)^{n(N-1)}
2^{-\beta'\! m(m+1)/2-\beta'\! (2m-1)(\lambda_{1}+\lambda_{2})/2+(2m-1)(1-2N+2l)+1}
(1-u)^{nl/2} \sqrt[4]{u} \end{align*}
and $$\theta_{N}=2(N-1)\theta  +(1+\beta'\!(m+\lambda_2))(\theta-\pi/4)+\beta'\!(\lambda_1-\lambda_2)\theta/2, \qquad \theta=\arccos \sqrt{u}.$$
\end{proof}

\begin{proof}[Proof of Corollary \ref{theobulkCF}: J$\beta$E case]
As previously, we suppose that $\beta$ is even,  which implies that  the scaling limits of the eigenvalue correlation functions in the bulk of the   J$\beta$E immediately follow  from  \eqref{evenbulkJE}. Let $n=2m=k\beta$, we have
\be R_{k,N}(x_1,\ldots,x_k)=\frac{(k+N)!}{N!}\frac{S_{N}(\lambda_{1},\lambda_{2},\beta/2)}{S_{k+N}(\lambda_{1},\lambda_{2},\beta/2)}   \prod_{1\leq j<l\leq
k}(x_{j}-x_{l})^{\beta}\,\big[\varphi_N(s_1,\ldots,s_n)\big]_{\{s\}\mapsto \{x\}}.
\ee
Notice $$\frac{S_{N}(\lambda_{1},\lambda_{2},\beta/2)}{S_{k+N}(\lambda_{1},\lambda_{2},\beta/2)} \sim \pi^{-k} (\Gamma(1+\beta/2))^{k} 2^{\beta(2N-1+k)k+2(\lambda_1+\lambda_2+1)k}(\beta N)^{-\beta k/2},$$
in   the bulk   taking  $\rho=\tfrac{1}{\pi}\tfrac{1}{\sqrt{u(1-u)}}$, one obtains
\begin{align} \big(\frac{1}{\rho N}\big)^{k} R_{k,N}\big(u+\frac{x}{\rho N}\big)\sim b_{k}(\beta) \, |\Delta(2\pi x)|^{\beta} {\!\!\!{\phantom{j}}_{0}\mathcal{F}_0^{( \beta/2)}}( i\pi s;1^{m}\!,(-1)^{m})_{\{s\}\mapsto \{x\}}
.\end{align}
The use of  Corollary \ref{practical} finally  allows us to rewrite the right-hand side in terms of a ${}_1F_1$  hypergeometric function as in Corollary \ref{theobulkCF}.
\end{proof}

At last we remark that  for odd $n=2m-1$ there is also a universal pattern in the bulk for all the three ensembles: H$\beta$E, L$\beta$E and J$\beta$E. Actually, from  the asymptotic results contained in Eqs \eqref{oddbulkHE}, \eqref{oddbulkLE} and\eqref{oddbulkJE},  one easily establishes the following theorem.

\begin{theorem}[``Kernel'' in the bulk] \label{odduniversality} Assume  that  $n=2m-1$ is odd. Let $A=\sqrt{2N},\, 4N,\, 1$ and  $\rho=\tfrac{2}{\pi}\sqrt{1-u^{2}}$, $\tfrac{2}{\pi}\sqrt{\tfrac{1-u}{u}}$ and $\tfrac{1}{\pi}\tfrac{1}{\sqrt{u(1-u)}}$
for  the $\HBE$, $\LBE$,  $\JBE$, respectively.  Moreover,  let  $s=(s_1,\ldots,s_n)$ and $t=(t_1,\ldots,t_n)$.  Then as $N\to \infty$,
\begin{align}\frac{1}{\Psi^{\!_{(0)}}_{\!^{N,2m-1}}\Psi^{\!_{(1)}}_{\!^{N,2m-1}}}&\Big\{
\varphi_N\big(A\,u+\frac{A\,s}{\rho N}\big)\varphi_{N-1}\big(A\,u+\frac{A\,t}{\rho N}\big)-\varphi_N\big(A\,u+\frac{A\,t}{\rho N}\big)\varphi_{N-1}\big(A\,u+\frac{A\,s}{\rho N}\big)
\Big\}\nonumber\\
&\sim \frac{1}{2i}\Big\{E_{m-1}^{(\beta/2\!)}( i\pi s)\, E_{m-1}^{(\beta/2\!)}( -i\pi t)-E_{m-1}^{(\beta/2\!)}( i\pi t)\, E_{m-1}^{(\beta/2\!)}( -i\pi s)
\Big\}\label{limitsinkernel}\end{align}
where   $E_k^{(\alpha)}$ denotes for the generalized exponential defined in \eqref{defE} and
$\Psi^{\!_{(l)}}_{\!^{N,2m-1}}$ is given in \eqref{eqPsiodd}.
\end{theorem}
The last theorem obviously generalizes the standard result   valid for the unitary (i.e., $\beta=2$)  ensembles and according to which  the polynomial kernel\footnote{Following our notation, the polynomial kernel $K_N(x,y)$ of the unitary ensembles is given by $\frac{\varphi_N(x)\varphi_{N-1}(y)- \varphi_N(y)\varphi_{N-1}(x)}{x-y}$  if $n=m=1$ and $\beta=2$.  }  $K_N(x,y)$  asymptotically tends to the sine kernel $  \frac{\sin \pi (x-y)}{\pi(x-y)}$ when it is rescaled in the bulk of the spectrum (see for instance \cite{forrester,mehta}).  We remark that whenever $n=m=1$,  the limiting kernel on the RHS of \eqref{limitsinkernel} is equal to $\sin \pi(s_1-t_1)$ no matter what the value of $\beta$ is.

\subsection{Parameter-varying ensembles} \label{parametervarying}

 As mentioned in the introduction, the Laguerre and Jacobi $\beta$-ensembles are defined in terms of additional parameters, respectively denoted by $\lambda_{1}$, and by $\lambda_{1}$ and $\lambda_{2}$.  Our previous results are proved only for fixed constants $\lambda_{1},\ \lambda_{2}\in (-1,\infty)$.

 However, in Statistics, these parameters are usually  written  as
\be \label{eqparametersvary} \lambda_{1}=\beta(N_{1}-N+1)/2-1, \qquad \lambda_{2}=\beta(N_{2}-N+1)/2-1, \ee
with $N_{1}, N_{2}>N$.  
 When  the parameters $N_{1}$ and $N_{2}$ vary proportionally with $N$, i.e. $$N_{j}/N \rightarrow \gamma_{j}\in [1,\infty),\qquad  (j=1,2),$$ one can still get the scaled limit for the empirical eigenvalue distribution (e.g. Chapter 3 in \cite{forrester}). It is worth noting that when $\gamma_{1}, \gamma_{2}>1$,  both the lower and upper edges are so-called  soft edges for these two ensembles, so the soft-edge behavior is expected.  Moreover, as shown in  \cite{hmf,jv}, the point process limits hold for the Laguerre and Jacobi $\beta$-ensembles with a general family of parameters, so the scaling limits for the averages of characteristic polynomials hold in these cases as well.

  \begin{theorem}[Soft-edge limit for parameter-varying ensembles]
\label{generaledgescaling}  Assume that $N_{1}/N \rightarrow
\gamma_{1}\in (1,\infty)$ for the $\LBE$ and $N_{1}/N \rightarrow
\gamma_{1}\in (1,\infty),\   N_{2}/N \rightarrow \gamma_{2}\in
[1,\infty)$ for the $\JBE$.  Moreover, let
\be  A,\, B=\begin{cases}4N\Lambda_{-}, \qquad 
-(\sqrt{\gamma_{1}}-1)^{4/3}\gamma_{1}^{-1/6}N^{1/3} & \text{for the
$\LBE$}, \\  b_{-}, \qquad    N^{-2/3}g_{-} & \text{for the $\JBE$,}
\end{cases}\ee
where 
\be \label{constantsvaryingcase}\Lambda_{\pm}=\big(\frac{1\pm\sqrt{\gamma_{1}}}{2}\big)^{2},\qquad  b_{\pm}=\Big(\frac{\sqrt{\gamma_{1}(\gamma_{1}+\gamma_{2}-1)}\pm
\sqrt{\gamma_{2}}}{\gamma_{1}+\gamma_{2}}\Big)^{2}\ee
and where  $g_{-}$ denotes the numerical coefficient  defined below in \eqref{Jacobisoftscaling}.
    Then ,      \eqref{eqlimitingsoft} in Theorem \ref{theosoft} and
\eqref{softuniversal} in Corollary \ref{theosoftCF} hold for the
$\LBE$ and $\JBE$.
\end{theorem}

 \begin{theorem}[Bulk limit for parameter-varying ensembles] \label{generaledgescaling2}  Let $\Lambda_{\pm}$ and $b_{\pm}$ be the constants defined in \eqref{constantsvaryingcase}.  For the $\LBE$, assume
that $A=4N$,  $N_{1}/N \rightarrow \gamma_{1}\in [1,\infty)$, and $$ \rho(u)=  \frac{2}{\pi u}\sqrt{(u-\Lambda_{-})(\Lambda_{+}-u)}, \qquad \Lambda_{-}<u<\Lambda_{+} .$$     For the $\JBE$, assume that $A=1$, 
$N_{i}/N \rightarrow \gamma_{j}\in [1,\infty)$ ($j=1,2$), and $$\rho(u)=
  \frac{ \gamma_{1}+\gamma_{2}}{2\pi
u(1-u)}\sqrt{(u-b_{-})(b_{+}-u)},\qquad  b_{-}<u<b_{+}. $$
Then,   \eqref{eqlimitingbulk}
in Theorem \ref{theobulk}, \eqref{bulkuniversal} in Corollary
\ref{theobulkCF} and \eqref{limitsinkernel} in Theorem
\ref{odduniversality} hold for the $\LBE$ and $\JBE$.
\end{theorem}

\begin{proof}[Sketch of Proof of Theorems \ref{generaledgescaling}  and \ref{generaledgescaling2} ]
One has to follow the same steps as in Sections 4.2 and 4.3.   Here, we only point out some slight differences that appear in the computation when the parameters vary as in \eqref{eqparametersvary}.   

 First of all, for the Laguerre case, let $N_{1}/N
\rightarrow \gamma_{1}\in [1,\infty)$ as $N \rightarrow \infty$.  Then
the saddle point equation  becomes
$ p'(x)=1/x+ \gamma_{1}/(1-x)-4u=0$
since $p(x)=\ln x-\gamma_{1}\ln(1-x)-4u x$ (corresponding to
\eqref{chLE3}). Hence, one gets the solutions
\be x_\pm=\frac{4u+1-\gamma_{1}}{8u}\pm
i\frac{1}{2u}\sqrt{(u-\Lambda_{-})(\Lambda_{+}-u)},\ee
 where \be\Lambda_{-}=\big(\frac{1-\sqrt{\gamma_{1}}}{2}\big)^{2},
\qquad \Lambda_{+}=\big(\frac{1+\sqrt{\gamma_{1}}}{2}\big)^{2}.\ee
 It shows that there are two simple saddle points in the bulk
($\Lambda_{-}<u<\Lambda_{+}$).  When $\gamma_{1}=1$, corresponding to the case where $N_{1}-N$ is fixed, one recovers  the hard-edge behavior which has been
studied in the previous sections.  When $\gamma_{1}>1$,  there are two
double saddle points respectively at   the left  and right edges of the
spectrum:
 \be x_{l}=\frac{1}{1-\sqrt{\gamma_{1}}},\qquad
x_{r}=\frac{1}{1+\sqrt{\gamma_{1}}}.\ee
 Note that in this case the density function  becomes
\be\rho(u)=\frac{2}{\pi u}\sqrt{(u-\Lambda_{-})(\Lambda_{+}-u)}.\ee

Now, for the Jacobi case, let $N_{j}/N \rightarrow \gamma_{j}\in
[1,\infty)$ ($j=1,2$) as $N \rightarrow \infty$. Since in this case, \be
p(x)=\gamma_{2}\ln x-(\gamma_{1}+\gamma_{2}-1)\ln(1-x)-\ln(1-u+u
x),\label{Jacobip}\ee   the saddle point equation becomes
$$ p'(x)=\gamma_{2}/x+ (\gamma_{1}+\gamma_{2}-1)/(1-x)-u/(1-u+u x)=0. $$
Its solutions are
\be x_\pm=\frac{(\gamma_{1}-\gamma_{2})u+1-\gamma_{1}\pm i\,
(\gamma_{1}+\gamma_{2})\sqrt{(u-b_{-})(b_{+}-u)}}{2\gamma_{1}u},\ee
 where \be b_{\pm}=\Big(\frac{\sqrt{\gamma_{1}(\gamma_{1}+\gamma_{2}-1)}\pm
\sqrt{\gamma_{2}}}{\gamma_{1}+\gamma_{2}}\Big)^{2}.\ee
   In particular, if one assumes $\gamma_{1}, \gamma_{2}>1$,  then one finds the following  double
saddle points  at   the left  and right edges of the spectrum:
 \be x_{l}=\sqrt{\frac{\gamma_{2}}{\gamma_{1}}}\frac{
\sqrt{\gamma_{1}\gamma_{2}}+ \sqrt{ \gamma_{1}+\gamma_{2}-1 }
 }{1- \gamma_{1} },\qquad
x_{r}=\sqrt{\frac{\gamma_{2}}{\gamma_{1}}}\frac{
\sqrt{\gamma_{1}\gamma_{2}}- \sqrt{ \gamma_{1}+\gamma_{2}-1 }
 }{1- \gamma_{1} }.\ee
The density function now needs to be replaced by \be \rho(u)=\frac{
\gamma_{1}+\gamma_{2}}{2\pi u(1-u)}\sqrt{(u-b_{-})(b_{+}-u)}.\ee
Note finally that the exact scaling constants at the soft edge of the Jacobi $\beta$-ensemble 
 with $\gamma_{1}, \gamma_{2}>1$, are given in terms of 
 \be \label{Jacobisoftscaling}g_{-}=(1-b_-+b_-
x_{l})^{2}\big(p'''(x_l)/2\big)^{1/3},\qquad g_{+}=(1-b_{+}+b_+
x_{r})^{2}\big(p'''(x_r)/2\big)^{1/3},\ee
 where $p'''(x)$ denotes the third derivative of $p(x)$ in \eqref{Jacobip}.
\end{proof}

\section{PDEs at the  edges and in the bulk}\label{pdes}

In the introduction, we presented the asymptotic values  of the rescaled expectations $\varphi_N(A+Bs)$ as $N\to\infty$.  Three multivariate special functions were used:  the multivariate hypergeometric functions  ${ }_{0} F_{1}(s_1,\ldots, s_n)$ (Bessel type) and ${ }_{1} F_{1}(s_1,\ldots, s_n)$ (Trigonometric type), respectively  at the hard edge and in the bulk, and the multivariate Airy function at the soft edge.  The first two can be written explicitly in terms of Jack polynomials while the Airy function is known via an $n$-dimensional integral formula involving the hypergeometric function of exponential type ${ }_{0}\mathcal{F}_{0}(s_1,\ldots, s_n;w_1,\ldots,w_n)$. In this short section, we show that these three types of limiting expectation values can be seen as solutions of three simple holonomic systems of $ n$ second-order  partial differential equations.  This will allow us to find, in the case where $n=2$, a very simple expression for the bulk limiting expectation.   For the soft edge, we also find a very simple limiting expectation in the case where $n$ is arbitrary, but $\beta=\infty$.

\begin{proposition}[Limiting PDEs]  Let $F$ respectively denote the scaling limit at the hard edge  on the RHS of \eqref{eqlimitinghard},  in the bulk on the RHS of \eqref{eqlimitingbulk} ($n=2m$), and at the soft edge \footnote{
By summing up all the $n$ equations in \eqref{softedge},  one easily shows that  the  multivariate Airy function is one possible  solution of  the following (single) PDE:
$ D_{0} F=p_{1}(s)F \,$
where $D_0$ is the differential operators defined in \eqref{diffops}. This PDE was first given \cite[Eq. (5.17)]{des} as an equation satisfied by the Airy function defined in \eqref{Airydef}.} on the RHS of \eqref{eqlimitingsoft}.  Then, $F$ satisfies respectively one system of PDEs
\begin{equation}\label{hardedge}%%\label{weightedhardedge} s_{k}\frac{\partial^{2}F}{\partial s^{2}_{k}}+ \frac{2}{\beta}\frac{\partial F}{\partial s_{k}}+\left(1-\frac{\lambda_{1}}{\beta}(\frac{\lambda_{1}+2}{\beta}-1)\frac{1}{s_{k}}\right)F+
%%\frac{2}{\beta}\sum_{j\neq k}
%%\frac{1}{s_{k}-s_{j}}\left(s_{k}\frac{\partial F}{\partial
%%s_{k}}-s_{j}\frac{\partial F}{\partial s_{j}}\right)=0
s_{k}\frac{\partial^{2}F}{\partial s^{2}_{k}}+ \frac{2}{\beta}(1+\lambda_{1})\frac{\partial F}{\partial s_{k}}+F+
\frac{2}{\beta}\sum_{j=1, j\neq k}^{n}
\frac{1}{s_{k}-s_{j}}\left(s_{k}\frac{\partial F}{\partial
s_{k}}-s_{j}\frac{\partial F}{\partial s_{j}}\right)=0 \ (k=1,\ldots,n),
\end{equation}
\begin{equation}\label{bulk} \frac{\partial^{2}F}{\partial s^{2}_{k}}+ F+\frac{2}{\beta}\sum_{j=1, j\neq k}^{n}
\frac{1}{s_{k}-s_{j}}\left(\frac{\partial F}{\partial
s_{k}}-\frac{\partial F}{\partial s_{j}}\right)=0  \ (k=1,\ldots,n),
\end{equation}
% %There  the author worked with the matrices themselves and   got the  asymptotic formula of same form for $\beta=1,2,4$, and further surmised that  it should be true for all $\beta$. The discussion above has actually answered that surmise.
\begin{equation}\label{softedge} \frac{\partial^{2}F}{\partial s^{2}_{k}}-s_{k}F+\frac{2}{\beta}\sum_{j=1, j\neq k}^{n}
\frac{1}{s_{k}-s_{j}}\left(\frac{\partial F}{\partial
s_{k}}-\frac{\partial F}{\partial s_{j}}\right)=0 \ (k=1,\ldots,n).
\end{equation}
\end{proposition}
\begin{proof}
 Since Kaneko \cite{kaneko} and Yan \cite{yan},  we know that the hypergeometric  function ${\phantom{|}}_{0} F_{1}^{(\beta/2)}((2/\beta)(\lambda_1+n);s)$ satisfies the following  holonomic system of PDEs:
\begin{equation} s_{k}\frac{\partial^{2}F}{\partial s^{2}_{k}}+ \frac{2}{\beta}(1+\lambda_{1})\frac{\partial F}{\partial s_{k}}-F+
\frac{2}{\beta}\sum_{j=1, j\neq k}^{n}
\frac{1}{s_{k}-s_{j}}\left(s_{k}\frac{\partial F}{\partial
s_{k}}-s_{j}\frac{\partial F}{\partial s_{j}}\right)=0\  (k=1,\ldots,n),
\end{equation}
   which implies    the first   system of PDEs given above.

The systems for the bulk and  the soft-edge of the classical $\beta$-ensembles were first observed in \cite{liu} by exploiting Kaneko's system of PDEs for the hypergeometric function  ${}_2F_1$ \cite{kaneko}.     However, the derivation in \cite{liu} was incomplete since it was assumed without proof that the scaling limits of $\varphi_N$ should exist in the bulk and at the edge.  Theorems \ref{theobulk} and \ref{theosoft} now make this assumption superfluous and the proposition follows.
\end{proof}

In the one-dimensional case, \eqref{bulk} and \eqref{softedge} respectively reduce to
the well-known differential equations satisfied by  the sine (or cosine)   and Airy functions.  However, in the  higher-dimensional case, it seems difficult to find all solutions of \eqref{bulk} and \eqref{softedge}, one of which being  our limiting expectation value.

Let us now pay more attention on the $n=2$ case.

\begin{proposition}[Bulk limit when $n=2$] Let $F(s_1,s_2)$ be the bulk limit on the RHS of \eqref{eqlimitingbulk}, then
\be \label{besselfunction} F(s_1,s_2)=2^{\frac{2}{\beta}-\frac{1}{2}}\Gamma(\frac{2}{\beta}+\frac{1}{2}) \,( s_1-s_2 )^{\frac{1}{2}-\frac{2}{\beta}}J_{\frac{2}{\beta}-\frac{1}{2}}( s_1-s_2 ).\ee
\end{proposition}

 \begin{proof} First of all, according to \eqref{eqlimitingbulk}, we have $ F(s_1,s_2)=\!\!\!{\phantom{j}}_0\mathcal{F}^{(\beta/2)}_0(i ,-i ;s_1,s_2)$. But  formula  \eqref{vip1} allows us to write
$$F(s_1,s_2)=e^{-i ( s_1-s_2)}\!\!\!{\phantom{j}}_0\mathcal{F}^{(\beta/2)}_0(2i ,0;s_1-s_2,0).$$

We see from the last expression that $F(s_1,s_2)$ only depends on $s_1-s_2$, so we set $F(s_1,s_2)=f(s_1-s_2)$.  Hence,  $f(x)$ is an analytic  function with $f(0)=1$ and $f(x)=f(-x)$.  Moreover, we know from \eqref{bulk}   that
$f(x)$ satisfies
\be f''+\frac{(4/\beta)}{x}f'+ f=0, \label{bulkoneeq}\ee
which can be reduced to the Bessel equation (cf. (4.5.9) in \cite{aar}). From this, we get
$$f(x)=2^{\frac{2}{\beta}-\frac{1}{2}}\Gamma(\frac{2}{\beta}+\frac{1}{2})\,   x ^{\frac{1}{2}-\frac{2}{\beta}}J_{\frac{2}{\beta}-\frac{1}{2}}( x),$$
where $J_{\alpha}(x)$ is the Bessel function of first kind of order $\alpha$, and the proposition follows. \end{proof}

It is worth noting that the RHS of \eqref{besselfunction} was first obtained by Aomoto \cite{aom2} (at zero) and by Su \cite{su} in the bulk of the $\HBE$ with $0<\beta<4$.  For the circular $\beta$-ensemble,  a similar result was obtained by Forrester\cite{forrester92} where Eq.\eqref{bulkoneeq} was also given (see Eq.(4.9) therein).

The case $n=2$ at the soft edge was also previously studied.  Indeed,  Su \cite{su}   proved for the  H$\beta$E,  with the aid of
 Dumitriu and  Edelman's tri-diagonal matrix model,  that the scaling limit of  $\varphi_N(s)$ has   a single integral   representation:
 $$\frac{1}{4\pi^{3/2}i}\int_{1-i\infty}^{1+i\infty} \frac{e^{\frac{1}{12}z^{3}-\frac{1}{2}(s_1+s_2)-\frac{1}{4z}(s_1-s_2)^{2}}}{z^{\frac{2}{\beta}+\frac{1}{2}}}d z.$$
 This integral was first  defined by K\"{o}sters \cite{koster}; it is believed to be proportional to our 2-dimensional integral representation.

We end this section with a few remarks on the $\beta=\infty$ case. %%, see a forthcoming paper on beta-ensembles with external sources \cite{dl2} for  more information.
  It is well known that the parameter $\beta$ of Random Matrix Theory  can be interpreted in Statistical Mechanics as the inverse temperature of a log-gas system.  Thus $\beta=\infty$ corresponds to the completely frozen state of the system, which is the state where the particles no longer move.  The following proposition  is  consistent with this  physical  phenomenon.

	\begin{proposition}[Soft-edge limit at $\beta=\infty$]
	Let $\mathrm{Ai}(x)$ denote  the one-variable Airy function of first kind. %($\mathrm{Bi}(x)$ for the second kind).
	Then
	\be \mathrm{Ai}^{(\infty)}(s)=\prod_{j=1}^{n} \mathrm{Ai}(s_j)\, .\ee
	\end{proposition}
	\begin{proof}
	As is well known
 $P^{(\infty)}_\kappa(s)=m_\kappa(s)$ \cite{stanley},  so we have
 \be \label{infinitehf}\!\!\!{\phantom{j}}_0\mathcal{F}^{(\infty)}_0(y;s)=\sum_{\kappa}\frac{1}{\kappa_1 !\cdots \kappa_n ! }\frac{m_{\kappa}(y)m_{\kappa}(s)}{m_{\kappa}(1^{n})}.\ee Now,  it is a simple exercise to show that
  $$\frac{m_{\kappa}(y)}{m_{\kappa}(1^{n})}=\frac{1}{n!}
 \sum_{\sigma\in S_{n}}\prod_{j=1}^{n}y_{j}^{\kappa_{\sigma(j)}}, $$
which readily implies  that
 \be \label{identity}\!\!\!{\phantom{j}}_0\mathcal{F}^{(\infty)}_0(y;s)=\frac{1}{n!}
 \sum_{\sigma\in S_{n}}\prod_{j=1}^{n}e^{s_{j}y_{\sigma(j)}}.
 \ee
After substitution of the latter result in  the integral formula \eqref{Airydef}, which defines the multivariate Airy function, one finally  gets the desired formula.
\end{proof}

 Note that the situation is a little   more complicated in the bulk.  For instance,
  when $n=2$,  use of \eqref{besselfunction} implies that
  \be \!\!\!{\phantom{j}}_0\mathcal{F}^{(\infty)}_0(i ,-i ;s_1,s_2)=2^{-\frac{1}{2}}\Gamma(\frac{1}{2})\sqrt{ s_1-s_2 }\,J_{-\frac{1}{2}}(s_1-s_2 )=\cos(s_1-s_2),\ee
  in which the variables cannot separate. We refer to \cite{due2} for more information on large $\beta$ asymptotics.

Note also that when $\beta=\infty$,  the systems of PDEs such as \eqref{bulk} and \eqref{softedge}  can be directly solved. Indeed, the linearly
independent symmetric solutions
are respectively \be \label{solutionssine}\frac{1}{n!}
 \sum_{\sigma\in S_{n}} e^{i (s_{\sigma(1)}+\cdots +s_{\sigma(j)}-s_{\sigma(j+1)}-\cdots -s_{\sigma(n)})}
 =\!\!\!{\phantom{j}}_0\mathcal{F}^{(\infty)}_0(1^{j},(-1)^{n-j};i s), \quad j=0,1, \ldots,n,\ee
 and
 \be \label{solutionsairy}\frac{1}{n!}
 \sum_{\sigma\in S_{n}}\mathrm{Ai}(s_{\sigma(1)})\cdots \mathrm{Ai}(s_{\sigma(j)})\mathrm{Bi}(s_{\sigma(j+1)})\cdots \mathrm{Bi}(s_{\sigma(n)}), \quad j=0,1, \ldots,n.\ee
 where   the equality in \eqref{solutionssine} comes from \eqref{identity} and $\mathrm{Bi}(x)$ denotes the one-variable Airy function of second kind.

Now a natural question arises: Can we guess results  similar
 to \eqref{solutionssine} and \eqref{solutionsairy} for general $\beta$? In particular,  are $\!\!\!{\phantom{j}}_0\mathcal{F}^{(\beta/2)}_0((-1)^{j},1^{n-j};i s)$ where $j=0,1, \ldots,n$ all  linearly
independent symmetric solutions of \eqref{bulk}? Note that  if $j=0$ or $n$ then  $$\!\!\!{\phantom{j}}_0\mathcal{F}^{(\beta/2)}_0((-1)^{j},1^{n-j};i s)=e^{\pm i p_{1}(s)}$$
satisfies \eqref{bulk}. Moreover, according to our bulk scaling at zero for the H$\beta$E, and for $n=2m$ or $2m-1$, the functions
$$\!\!\!{\phantom{j}}_0\mathcal{F}^{(\beta/2)}_0((-1)^{m},1^{n-m};\pm i s)$$
are also solutions of \eqref{bulk} (see \eqref{evenbulkHE} and \eqref{oddbulkHE} for $\theta=0,\, l=0,1$).

\textbf{Conjecture}:\textit{ The set of symmetric solutions of the system of PDEs \eqref{bulk} is spanned  by the  $n+1$ linearly independent functions 
$$\!\!\!{\phantom{j}}_0\mathcal{F}^{(\beta/2)}_0((-1)^{j},1^{n-j};i s), \qquad  j=0,1, \ldots,n. $$}

\begin{acknow}

The work of P.~D.\ was  supported by FONDECYT grant \#1090034 and by CONICYT through the Anillo de Investigaci\'on ACT56.
The work of D.-Z.~L.\ was  supported by FONDECYT grant \#3110108 and partially by the National Natural Science Foundation of China (Grant No. 11171005).
\end{acknow}

 \begin{appendix}

\section{Notation and  constants}

Most of the constants used in the article can be derived from Selberg's integrals \cite{forrester,mehta}:
\begin{align} \label{Selbergformula} S_{N}(\lambda_{1},\lambda_{2},\lambda_3):&= \int_{[0,1]^N}\prod_{i=1}^N x_i^{\lambda_1}(1-x_i)^{\lambda_2}\prod_{1\leq k<j \leq N}|x_k-x_j|^{2\lambda_3}\,d^{N}x\nonumber\\
&=\prod_{j=0}^{N-1}
\frac{\Gamma(1+\lambda_3+j\lambda_3) \Gamma(1+\lambda_{1}+j\lambda_3) \Gamma(1+\lambda_{2}+j\lambda_3)}
{\Gamma(1+\lambda_3)\Gamma(2+\lambda_{1}+\lambda_{2}+(N+j-1)\lambda_3)}.
\end{align}
In particular, one readily shows that
\be \label{constantforl}W_{\lambda_{1}, \beta,N}=(2/\beta)^{(1+\lambda_{1})N+\beta N(N-1)/2}\prod_{j=0}^{N-1}
\frac{\Gamma(1+\beta/2+j\beta/2) \Gamma(1+\lambda_{1}+j\beta/2)}
{\Gamma(1+\beta/2)},\ee
 and \be \label{constantforg} G_{\beta,N}=\beta^{-N/2-\beta N(N-1)/4}(2\pi)^{N/2}\prod_{j=0}^{N-1}
\frac{\Gamma(1+\beta/2+j\beta/2)}
{\Gamma(1+\beta/2)}.\ee
For the Gaussian case, it is often more convenient to use the following integral:
\be\label{gaussintgen} \int_{\mathbb{R}^n}\prod_{i=1}^n e^{-zx_i^2/2}\prod_{1\leq k<j \leq n}|x_k-x_j|^{\beta}\,d^{n}x=\frac{1}{z^{(n+\beta n(n-1)/2)/2}}\Gamma_{\beta,n},\qquad \textrm{Re}\{z\}>0,
\ee
where
\be \label{constgamma}
 \Gamma_{\beta,n}=(2\pi)^{n/2}\prod_{j=1}^n\frac{\Gamma(1+j\beta/2)}{\Gamma(1+\beta/2)}.
\ee
%%and
%%\be\label{notationsub} n_{\beta}= \beta n(n-1)/2.
%%\ee
The Morris normalization constant is
\be M_n(a,b, \alpha)= \prod_{j=0}^{n-1}
\frac{\Gamma(1+ \alpha+j \alpha) \Gamma(1+a+b+j \alpha)}
{\Gamma(1+ \alpha)\Gamma(1+a+j \alpha)\Gamma(1+b+j \alpha)}.\ee
The constants for the  soft edge are
\be\label{eqPhi}
\Phi_{\!^{N,n}}=\begin{cases}  N^{\beta'\! n(n-1)/12+n/6}\exp\{ -n N(1+\ln 2-\ln N-2i\pi)/2\}  &\HBE \\
   2^{-\beta'\! n(n-1)/6-\beta'\! n/2+2n/3} N^{\beta'\! n(n-1)/12+\beta'\! n \lambda_1/4+ n/6}\exp\{ -n N(1-\ln N-i\pi)\}&\LBE
\end{cases}
\ee
In the bulk with $n=2m$,  the constants are
\be\label{eqPsi}
\Psi_{\!^{N,2m}}=\begin{cases}(\pi \rho)^{\beta'\! m(m+1)/2-m} N^{\beta'\! m^{2}/2}\exp\{ -m N(1+\ln 2-\ln N)\} &\HBE \\
(\pi \rho/2)^{\beta'\! m(m+1)/2-m} N^{\beta'\! m(m+\lambda_1)/2}\exp\{ -2m N(1-\ln N)\} &\LBE \\
 (\pi \rho)^{\beta'\! m(m+1)/2-m} N^{\beta'\! m^{2}/2}2^{-\beta'\! m^{2}/2-\beta'\! m(\lambda_{1}+\lambda_{2}+1)+2m(1-2N)} & \JBE
\end{cases}
\ee where  $\rho=\tfrac{2}{\pi}\sqrt{1-u^{2}}$, $\tfrac{2}{\pi}\sqrt{\tfrac{1-u}{u}}$, $\tfrac{1}{\pi}\tfrac{1}{\sqrt{u(1-u)}}$,
respectively correspond to the $\HBE$, $\LBE$ and $\JBE$. \\
The constants for the bulk with $n=2m-1$ are
\be\label{eqPsiodd}
\Psi_{\!^{N,2m-1}}^{_{(l)}}=\begin{cases}&\tbinom{2m-1}{m} \Gamma_{\beta'\!,m-1} \Gamma_{\beta'\!,m}(\Gamma_{\beta'\!,2m-1})^{-1}(\pi \rho)^{\beta'\!(m^{2}-1)/2-(2m-1)/2} N^{\beta'\! m(m-1)/2} \\ &\times   \exp\{ -(2m-1) N(1+\ln 2-\ln N)/2\}\, (\sqrt{N/2})^{-(2m-1)l} (2i\sqrt{\pi \rho/2}) \hspace{2.28cm}\,\HBE \\
&\tbinom{2m-1}{m} \Gamma_{\beta'\!,m-1} \Gamma_{\beta'\!,m}(\Gamma_{\beta'\!,2m-1})^{-1}(\pi \rho/2)^{\beta'\!(m^{2}-1)/2- m+1} \sqrt{2}(2\sqrt{u})^{-\beta'\!/2+1}  \\ &\times  N^{\beta'\! m(m-1)/2+\beta'\! (2m-1)\lambda_1/4} \exp\{ -(2m-1) N(1-\ln N-i\pi)\}\, (-N)^{-(2m-1)l} \hspace{1.cm}  \LBE \\
& \tbinom{2m-1}{m} \Gamma_{\beta'\!,m-1} \Gamma_{\beta'\!,m}(\Gamma_{\beta'\!,2m-1})^{-1}(\pi \rho)^{\beta'\!(m-1)^{2}/2 +(\beta'\!-2m+1)/2} \sqrt{2}(2\sqrt{u})^{-\beta'\!/2+1} N^{\beta'\! m(m-1)/2}  \\ &\times  i^{-1}(-1)^{n(N-1)}
2^{-\beta'\! m(m+1)/2-\beta'\! (2m-1)(\lambda_{1}+\lambda_{2})/2+(2m-1)(1-2N+2l)+1}
(1-u)^{nl/2} \sqrt[4]{u}  \hspace{.3cm}\JBE
\end{cases}
\ee
For the hard edge, the constants are
\be\label{eqxi}
\xi_{N,n}=\begin{cases} \frac{W_{\lambda_1+n, \beta,N}}{W_{\lambda_1, \beta,N}} &\LBE \\  \frac{S_N(\lambda_1+n,\lambda_2;\beta/2)}{S_N(\lambda_1,\lambda_2;\beta/2)} & \JBE.
\end{cases}
\ee
Finally, the universal coefficients  are \be a_{k}(\beta)=  (\beta/2)^{(\beta k+1)k}(\Gamma(1+\beta/2))^{k} \prod_{j=1}^{2k} \frac{(\Gamma(1+2/\beta))^{\beta /2}}{\Gamma(1+\beta j/2) }\label{univcoefficienta}, \ee
  \be b_{k}(\beta)=  (\beta/2)^{\beta k(k-1)/2}(\Gamma(1+\beta/2))^{k} \prod_{j=0}^{k-1}\frac{\Gamma(1+\beta j/2)}{\Gamma(1+\beta (k+j)/2)} \label{univcoefficientb},\ee
and \be \gamma_{m}(\beta'\!)=\binom{2m}{m}\prod_{j=1}^m\frac{\Gamma(1+\beta'\!j/2)}{\Gamma(1+\beta'\!(m+j)/2)}.\label{univcoefficientevenbulk}\ee
 \end{appendix}


\begin{thebibliography}{99}






\bibitem{af} G.\ Akemann and Y.\ V.\ Fyodorov, \textit{Universal random matrix correlations of ratios of characteristic polynomials at the spectral edges}, Nucl.\ Phys.\ B, 664 (2003), 457--476.

\bibitem{aar} G.\ E.\ Andrews, R. Askey, R. Roy, \textit{Special Functions}, Cambridge University Press (2000).

\bibitem{aom2} K. Aomoto, \textit{Scaling limit formula for 2-point correlation function of random matrices},   Adv. Stud. Pure Math. 16 (1988): conformal field theory and solvable lattice models, 1--15.

\bibitem{bf} T.~H.~Baker and P.~J.~Forrester, \emph{The Calogero-Sutherland model and generalized classical polynomials}, Commun.~Math.~Phys.~188 (1997), 175--216.

\bibitem{bds} J.\ Baik, P.\ Deift, and E.\ Strahov, \textit{Products and ratios of characteristic polynomials of random Hermitian matrices}, J.\ Math.\ Phys.\ 44 (2003), 3657--3670.




\bibitem{bempf} M.\ Berg\`ere, B.\ Eynard, O.\ Marchal, A.\ Prats-Ferrer, \emph{Loop equations and topological recursion for the arbitrary-$\beta$ two-matrix model}, JHEP 03 (2012) 098, arXiv:1106.0332.



\bibitem{bs} A.\ Borodin and E.\ Strahov, \textit{Averages of Characteristic Polynomials in Random Matrix Theory},  Commun.\ Pure and Appl.\ Math.\ 59 (2006),  161--253.

\bibitem{bey} P.\ Bourgade, L.\ Erd\"os, and  H.-T.\ Yau,  \textit{Universality of General $\beta$-Ensembles}, arXiv:1104.2272v5 (2012), 48 pp.

\bibitem{bey2} P.\ Bourgade, L.\ Erd\"os, and  H.-T.\ Yau, \textit{Bulk Universality of General $\beta$-Ensembles with Non-convex Potential}, arXiv:
1201.2283v2 (2012), 22 pp.

\bibitem{breitung}K.\ Breitung and M.\ Hohenbichler, \textit{ Asymptotic approximations for multivariate integrals with an application to multinormal probabilities}, J.\ Multivar. Analysis 30 (1989), 80--97.

\bibitem{bh} E.\ Br\'ezin and S.\ Hikami, \textit{ Characteristic Polynomials of Random Matrices}, Commun.\ Math.\  Phys.\ 214 (2000), 111--135.

\bibitem{caer} G.\ Le Ca\"er, C.\ Male, and R.\ Delannay, \emph{Nearest-neighbour spacing distributions of the $\beta$-Hermite}, Physica A 383 (2007), 190--208.


\bibitem{cem}	
L.\ O.\ Chekhov, B.\ Eynard, O.\ Marchal, \textit{Topological expansion of the $\beta$-ensemble model and quantum algebraic geometry in the sectorwise approach},
Theoretical and Mathematical Physics 166 (2011), 141--185.

\bibitem{des} P.\ Desrosiers, \textit{Duality in random matrix ensembles for all $\beta$}, Nucl.\ Phys.\ B 817 (2009), 224--251.




\bibitem{df} P.\ Desrosiers and P. J.\ Forrester, \textit{Hermite and Laguerre $\beta$-ensembles: Asymptotic corrections to the eigenvalue density}, Nucl.\ Phys.\ B 43 (2006), 307--332.

%\bibitem{dh} P. Desrosiers, M. Halln\"{a}s,  \textit{Hermite and Laguerre symmetric functions associated with operators of Calogero-Moser-Sutherland type}, arXiv:1103-4593v1.

\bibitem{dl} P. Desrosiers, D.-Z. Liu, \textit{Selberg Integrals, super hypergeometric functions and applications to $\beta$-ensembles of random matrices}, arXiv:1109.4659, 43 pages.
%\bibitem{dl2} P. Desrosiers and D.-Z. Liu, in preparation.



\bibitem{dv} L.\ Dumaz and  B. Vir\'ag,\textit{ The right tail exponent of the Tracy-Widom-beta distribution}, arXiv:1102.4818 (2011), 24 pp.

\bibitem{due} I.\ Dumitriu and A.\ Edelman, \textit{	
Matrix models for beta ensembles}, J.\ Math.\ Phys.\ 43 (2002), 5830--5847.
\bibitem{due2} I.\ Dumitriu and A.\ Edelman, \textit{	
Eigenvalues of Hermite and Laguerre ensembles:
large beta asymptotics}, Ann.\ I. H. P. Prob. Stat.\ 41 (2005), 1083--1099.
\bibitem{due3} I.\ Dumitriu and A.\ Edelman, \textit{Global spectrum fluctuations for the $\beta$-Hermite and $\beta$-Laguerre ensembles via matrix models}, J.\ Math.\ Phys.\ 47 (2006), 063302 1--36.

\bibitem{duk}  I.\ Dumitriu and  P.\ Koev, \textit{Distributions of the extreme eigenvalues of beta-Jacobi random matrices}, SIAM J.\   Matrix Anal.\    Appl.\ 30 (2008), 1--6.




\bibitem{es} A.\ Edelman and B.\ D.\ Sutton, \textit{From Random Matrices to Stochastic Operators}, J.\  Stat.\ Phys.\ 127 (2007),  1121--1165.


\bibitem{ER} A.\ Edelman, N.\ R.\ Rao,  \textit{Random matrix theory}, Acta Numerica 14 (2005), 233--297.

%\bibitem{erdos}  L.\ Erd\"os and H.-T.\ Yau, \textit{Universality of local spectral statistics of random matrices},  Bull.\ Amer.\ Math.\ Soc.\ 49 (2012), 377--414.


\bibitem{forrester92}P.\ J.\  Forrester, \textit{Selberg correlation integrals and the $1/r^{2}$quantum
many-body system}, Nucl.\ Phys.\ B 388 (1992),
671--699.
\bibitem{foredge}P.\ J.\  Forrester, \textit{The spectrum edge of random matrix ensembles}, Nucl.\ Phys.\ B 402 (1993),
709--728.

\bibitem{forrester0} P.\ J.\    Forrester, \textit{Exact results and universal asymptotics in the Laguerre random matrix ensemble}, J.\ Math.\ Phys.\ 35 (1993), 2539--2551.
\bibitem{forrester1} P.\ J.\    Forrester, Beta random matrix ensembles, \textit{Random matrix theory and its applications}(Z. Bai, Y. Chen and Y.-C. Liang eds.),
Lecture Notes Series, IMS, NUS, vol. 18, World Scientic, Singapore, 2009, pp.27--68.

\bibitem{forrester} P.\ J.\    Forrester, \textit{Log-gases and Random Matrices}, London Mathematical Society Monographs 34, Princeton University Press (2010).



\bibitem{for2012}P.\ J.\   Forrester, \textit{ The averaged characteristic polynomial for the Gaussian and chiral Gaussian ensembles with a source}, 	 arXiv:1203.5838v1, 21 pages.


\bibitem{FS} P.\ J.\ Forrester, M.\ J.\ Sorrell, \emph{Asymptotics of spacing distributions 50 years later}, 	arXiv:1204.3225v2, 21 pages.


%\bibitem{fw}P. J. Forrester, S.O. Warnaar, \textit{The importance of the Selberg integral}, Bull.\ Amer.\ Math.\ Soc.\ 45 (2008), 489--
%534.

%\bibitem{fyodorov}
%Y.\ Fyodorov, \textit{Random matrix theory},  Scholarpedia (2011) 6(3):9886.

\bibitem{grad}I. S. Gradshteyn,  I. M. Ryzhik, \textit{Table of integrals, series, and products}, Academic Press, 7th edition, 2007.


\bibitem{hmf}
D.\ Holcomb, G.\ R.\ Moreno Flores, \textit{Edge scaling of the $\beta$-Jacobi ensemble}, J. Stat. Phys. 149 (2012), 1136--1160.


\bibitem{hsu} L.\ C.\ Hsu, \textit{A theorem on the asymptotic behavior of a multiple integral},  Duke Math.\ J.\ Vol. 15 (1948), 623--632.

 \bibitem{jv} S.\ Jacquot and  B. Valk\'o, \textit{ Bulk scaling limit of the Laguerre ensemble}, EJP 16 (2011), 314--346.

\bibitem{kadell} K.\ W.\ J.\ Kadell, \textit{The Selberg-Jack symmetric functions}, Adv.\ Math.\ {130} (1997), 33--102.

\bibitem{kaneko} J.\ Kaneko, \textit{Selberg integrals and hypergeometric functions associated with Jack polynomials},
SIAM  J.\ Math.\ Anal.\ {24} (1993), 1086--1110.

\bibitem{kato}
 Y.\ Kuramoto and Y.\ Kato, \textit{Dynamics of one-dimensional quantum systems}, CUP, Cambridge, 2009.

\bibitem{kill}
R.~Killip, \textit{Gaussian Fluctuations for $\beta$ Ensembles}, IMRN (2008) rnn007, 1--19.

\bibitem{KN}
R.~Killip and I.\ Nenciu, \textit{Matrix Models for Circular Ensembles}, IMRN (2004), 2665--2701.

\bibitem{ks} J. P. Keating and N. C. Snaith, \textit{Random matrix theory and $\zeta (1/2 +
it)$}, Commun. Math. Phys. 214 (1) (2000), 57--89.

\bibitem{ke} P.\ Koev and A.\ Edelman, \textit{The efficient evaluation of the hypergeometric function of a matrix argument}, Mathematics of Computation 75 (2006), 833--846.

\bibitem{kon}
M.\ Kontsevich, \textit{Intersection theory on the moduli space of curves and the matrix Airy
function}, Commun.\ Math.\ Phys.\ 147 (1992), 1--23.

\bibitem{koranyi} A. Kor\'{a}nyi,\textit{ Hua-type integrals, hypergeometric functions and symmetric polynomials}, in: International Symposium in Memory of Hua Loo Keng, Beijing, 1988, vol. II, Springer, Berlin, 1991, 169--180.


\bibitem{koster} H. K\"{o}sters,  Asymptotics of characteristic polynomials of Wigner matrices at the edge of the spectrum,
Asymptotic Analysis, Vol.69(3-4) (2010), 233--248.
\bibitem{liu} D.-Z. Liu,  PhD thesis (in Chinese), Peking University, 2010.

\bibitem{macdonald} I.\ G.\ Macdonald, \textit{Symmetric Functions and Hall Polynomials} (2nd
ed.), Oxford University Press Inc, New York, 1995.

\bibitem{matsumoto1} S.\ Matsumoto, \textit{ Moments of characteristic polynomials for compact symmetric spaces and Jack polynomials}, J.\ Phys.\ A 40 (2007), 13567--13586.


\bibitem{matsumoto2} S.\ Matsumoto, \textit{ Jack deformations of Plancherel measures and traceless Gaussian random
matrices}, Electronic J.\ of Combinat.\ 15 (2008), article R149, 18 pages.

\bibitem{matsumoto3} S.\ Matsumoto, \textit{  Jucys-Murphy elements, orthogonal matrix integrals, and Jack measures}, The
Ramanujan Journal 26 (2011), 69--107.



\bibitem{mehta} M.\ L.\ Mehta, \textit{Random Matrices}, 3rd ed., Elsevier Academic Press
(2004).


\bibitem{mmm} A.\ Mironov, A.\ Morozov, and A.\ Morozov, \textit{Conformal blocks and generalized Selberg integrals},  Nucl.\ Phys.\ B 843 (2011), 534--557.

\bibitem{nagao} T.\ Nagao, P.\ J.\ Forrester, \textit{Asymptotic correlations at the spectrum edge of random matrices}, Nucl.\ Phys.\ B 435 (1995), 401--420.

\bibitem{olver}  F.W.J. Olver, \textit{Asymptotics and special functions}. AKP Classics, Wellesley, MA: A K Peters Ltd., 1997.
Reprint of the 1974 original. New York: Academic Press.

\bibitem{rr}  J.\ A.\ Ram\'irez, B.\ Rider, \textit{Diffusion at the Random Matrix Hard Edge}, Commun.\ Math.\  Phys.\ 288 (2009), 887--906.

\bibitem{rrv} J.\ A.\ Ram\'irez, B.\ Rider, and B.\ Vir\'ag, \textit{Beta ensembles, stochastic Airy spectrum, and a diffusion}, J.\ Amer.\ Math.\ Soc.\ 24 (2011), 919--944.

%\bibitem{selberg} A. Selberg, \textit{ Bemerkninger om et multipelt integral}, Norsk. Mat. Tidsskr. 24 (1944), 71--78.

\bibitem{rrz} J.\ A.\ Ram\'irez, B.\ Rider,  and O. Zeitouni, \textit{ Hard edge tail asymptotics}, ECP 16 (2011), 741--752.


\bibitem{su} Z. G. Su, \textit{On the second-order correlation of characteristic polynomials of
Hermite $\beta$ ensembles}, Stat. Proba. Letters. 80 (2010), 1500--1507.

\bibitem{sulkovski} P.\ Su{\l}kowski, \textit{Matrix models for $\beta$-ensembles from Nekrasov partition functions}, Journal of High Energy Physics (2010) 063.1--063.36.

\bibitem{stanley} R.~P.~Stanley, \emph{Some combinatorial properties of Jack symmetric functions}, Adv.~Math.~77 (1989), 76--115.


\bibitem{vv} B.\ Valk\'o and B.\ Vir\'ag, \textit{Continuum limits of random matrices and the Brownian carousel}, Inventiones Mathematicae Vol. 177 (2009), 463--508.


\bibitem{vv2} B.\ Valk\'o and B.\ Vir\'ag,  \textit{Large gaps between random eigenvalues}, Annals of Probab.\ 38 (2010), 1263--1279.
\bibitem{yan} Z. Yan, \textit{ A class of generalized hypergeometric
functions in several variables}, Can.\  J.\ Math.\ 44 (1992), 1317--1338.

\end{thebibliography}
\end{document}